\documentclass[
reprint,
superscriptaddress,
 amsmath,amssymb,
 aps,
 prx,
longbibliography,
floatfix,
]{revtex4-2}

\usepackage{xcolor}
\usepackage{soul}
\definecolor{mylightblue}{HTML}{1172ba}
\definecolor{mygray}{HTML}{989898}
\definecolor{myorange}{HTML}{f15a29}
\definecolor{mygreen}{HTML}{009444}
\usepackage{graphicx}
\usepackage{dcolumn}
\usepackage{bm}
\usepackage{braket}
\usepackage{amsthm}
\usepackage{tikz-cd}
\usepackage{thm-restate}
\usetikzlibrary{cd,fit}
\usepackage{hyperref}
\usepackage{bbold}
\hypersetup{
    colorlinks=true,
    linkcolor=blue,
    urlcolor=blue, 
    citecolor=blue,
}
\usepackage{nicematrix, siunitx}

\newtheorem{definition}{Definition}
\newtheorem{remark}{Remark}
\newtheorem{theorem}{Theorem}

\newtheorem{proposition}{Proposition}
\newtheorem{lemma}{Lemma}

\newcommand*\bigcdot{\mathpalette\bigcdot@{.5}}

\newcommand{\cA}[0]{A_\bullet}
\newcommand{\cAi}[1]{A[#1]_\bullet}
\newcommand{\cG}{G_\bullet}
\newcommand{\cGi}[1]{G[#1]_\bullet}
\newcommand{\cC}[0]{C_\bullet}
\newcommand{\coA}[0]{A^\bullet}
\newcommand{\coC}[0]{C^\bullet}
\newcommand{\cD}[0]{D_\bullet}
\newcommand{\coD}[0]{D^\bullet}
\newcommand{\im}[0]{\mathrm{im}}
\newcommand{\id}[0]{\mathrm{id}}
\newcommand{\cone}[0]{\mathrm{cone}}
\newcommand{\bF}[0]{\mathbb{F}}
\newcommand{\cc}[0]{\mathtt{CC}}

\newcommand{\Span}{\mathrm{span}}

\newcommand\IBM{
    IBM Quantum, 
    IBM T.~J. Watson Research Center, 
    Yorktown Heights, NY 10598, USA
    }

\newcommand\MIT{
    Department of Mathematics, 
    Massachusetts Institute of Technology, 
    Cambridge, MA 02139, USA
    }
    
\newcommand\Yale{
    Department of  Physics, 
    Yale University, 
    New Haven, CT 06520, USA
    }
    
\newcommand\YQI{
    Yale Quantum Institute, 
    Yale University, 
    New Haven, CT 06511, USA
    }

\begin{document}

\title{Constant-Time Surgery on 2D Hypergraph Product Codes \\
with Near-Constant Space Overhead}

\author{Kathleen~(Katie)~Chang}
\altaffiliation{These authors contributed equally to this work.}
\affiliation{\IBM}
\affiliation{\Yale}
\affiliation{\YQI}

\author{Zhiyang~He}
\altaffiliation{These authors contributed equally to this work.}
\affiliation{\MIT}

\author{Theodore~J.~Yoder}
\affiliation{\IBM}

\author{Guanyu~Zhu}
\affiliation{\IBM}

\author{Tomas~Jochym-O'Connor}
\affiliation{\IBM}

\begin{abstract}
    Generalized code surgery is a versatile and low-overhead technique for performing fault-tolerant computation on quantum low-density parity-check (qLDPC) codes. 
    In many settings, surgery exhibits practical space overheads, while its time overhead remains a bottleneck at $O(d)$ syndrome rounds per operation.
    In this work, we construct surgery gadgets that perform parallel logical measurements on 2D hypergraph product codes in constant time overhead ($O(1)$) and near-constant space overhead ($\tilde{O}(1)$).
    The reduced time overhead is a result of amortization, as we show, following the formulation by Cowtan et al. (\href{https://arxiv.org/abs/2510.14895}{arXiv:2510.14895}), that performing $d$ surgery operations in $O(d)$ time is fault tolerant.
    Our gadgets combine the strengths of different approaches to fault-tolerant logical operations: they partially retain the flexibility of surgery while achieving overheads comparable to transversal gates.
    Consequently, they are well-suited for near-term experimental realization and demonstrate new possibilities in the design of gadgets for fast logical computation.
\end{abstract}

\maketitle

\section{Introduction}

A central goal in the field of quantum error correction (QEC) is to design fault-tolerant quantum computing (FTQC) systems with low space- and time overhead. 
In recent years, progress in the study of quantum low-density parity-check (qLDPC) 
codes have demonstrated that by using high-rate qLDPC codes~\cite{gottesman2014fault-tolerant,bravyi2024high-threshold2,xu2024constant-overhead}, the asymptotic and practical space overhead of implementing fault-tolerant quantum memory can be significantly reduced compared to using low-rate topological codes~\cite{dennis2002topological}. 
Such progress motivated research into methods of performing fault-tolerant logical computation on qLDPC memories, which has since became a popular and active area of research.

A versatile and low-overhead technique to compute with qLDPC codes is code surgery, which is a generalization of lattice surgery~\cite{horsman2012surface} and has seen rapid developments in the past few years~\cite{cohen2022low-overhead,cowtan2024ssip,cross2024improved,williamson2024low-overhead,Ide_2025,swaroop2024universal,zhang2025time-efficient,cowtan2025parallel,he2025extractors,yoder2025tour,zheng2025high,zhang2025accelerating,baspin2025fast,cowtan2025fast}.
We refer readers to Section~3.2 of Ref.~\cite{he2025extractors} for a high-level review of recent progress, and supplement it by Section~\ref{sec:intro_fast_surgery} of this work.
Computationally, code surgery implements fault-tolerant measurements of logical Pauli operators on stabilizer codes. 
Such measurements can be used to implement Pauli $\pi/4$- and $\pi/8$-rotations on logical qubits (when assisted by magic states), which are sufficient for universal quantum computation.
This model of universal computation is known as Pauli-based computation (PBC)~\cite{bravyi2016trading}, which is implemented by modern lattice surgery-based surface code architectures~\cite{litinski2019game,gidney2025factor} as well as surgery-based qLDPC code architectures under the umbrella of extractor architectures~\cite{he2025extractors,yoder2025tour,webster2026pinnacle}.

Procedurally, surgery introduces an ancilla system of qubits and stabilizers to the original code (also called the \textit{base code}), which deforms it into a \textit{deformed code}. 
The ancilla system is designed such that the operators to be measured are products of newly introduced low-weight stabilizers.
As a result, upon learning the outcomes of the stabilizers in the deformed code reliably, one can deduce the measurement of the logical Pauli. 
Generically, learning the value of these stabilizers takes $O(d)$ rounds of measurement in the presence of noise, where $d$ is the code distance.
This non-trivial time overhead motivates a natural line of inquiry: under what assumptions can we perform surgery with fewer rounds of measurements, and can this time overhead be reduced to constant asymptotically?
To answer these questions, past works have studied two different forms of fast surgery, which we distinguish as \textit{single-shot surgery} and \textit{constant-time surgery}. 
We discuss their differences and the constructions from prior works in Section~\ref{sec:intro_fast_surgery}.

In this work, we introduce constant time overhead surgery gadgets for arbitrary two-dimensional hypergraph product (HGP) codes~\cite{tillich2014quantum,eczoo_hypergraph_product} which enable flexible and parallel measurements of logical Pauli operators. 
Our gadgets cost $\tilde{O}(n)$ ancilla qubits each, where $\tilde{O}(\cdot)$ denotes omission of poly-logarithmic factors. 
Together, the multiplicative spacetime overhead is $\tilde{O}(1)$, which we expect to be a small constant in practice.

Our construction features two conceptual novelties relative to the rich existing literature on qLDPC computation and code surgery. 
First of all, prior works~\cite{hillmann2024single-shot,Tan2025Single} have observed that constant spacetime overhead surgery can be performed on three- and higher-dimensional HGP codes. 
Notably, these codes are known to admit single-shot state preparation, which means their stabilizer checks in the $X$ (or $Z$) basis have enough redundancy such that they can be reliably measured in $O(1)$ time. 
This powerful property largely implies the plausibility of single-shot surgery~\cite{hillmann2024single-shot}. 
In contrast, generic two-dimensional HGP codes do not admit such redundancy in their stabilizer checks. 
In fact, most families of 2D HGP codes (including the surface code) require $O(d)$ rounds of syndrome measurements to reliably infer stabilizer values. 
From this perspective, our construction of constant-time surgery on 2D HGP codes is perhaps surprising. 
We discuss the connection to algorithmic fault-tolerance~\cite{zhou2024algorithmic} in Section~\ref{sec:intro_fast_surgery}.

More broadly speaking, among the proposed schemes for qLDPC computation, surgery is known for being highly versatile (capable of performing a large set of gates), while constant-depth logical gates are favored for their minimal overhead. 
Our construction is a rare instance which achieves close to the best of both worlds: it partially inherits the addressability from surgery, at roughly the same asymptotic cost of transversal gates. 
Moreover, the construction applies to all 2D HGP codes, and can be generalized to higher-dimensional HGP codes as well. 
As an example, we derive a detailed case study on the toric code, where our surgery gadgets have multiplicative constant spacetime overhead. 
We further expect our construction to be extendable to balanced product~\cite{breuckmann2021balanced} and lifted product codes~\cite{panteleev2021degenerate,panteleev2022asymptotically}

The rest of this paper is organized as follows. 
We first review existing formulations of fast surgery at a high-level in Section~\ref{sec:intro_fast_surgery}. 
In Section~\ref{sec:summary_construction}, we summarize our main construction for fast surgery gadgets and walk through an explicit, small example of our construction. 
We briefly review prior works in Section~\ref{sec:previous-works}.
We present the necessary technical background for understanding existing surgery techniques in Section~\ref{sec:prelims}. 
In Section~\ref{sec:hgp-code}, we present our main construction on 2D HGP codes and prove its properties. 
In Section~\ref{sec:examples_applications}, we discuss various ways to apply our construction, such as entangling two different HGP codes that share one common component code. 
In Section~\ref{sec:toric-code}, we present an intuitive and detailed case study on the 2D toric code that achieves constant space and time overhead.
Finally, we discuss some future directions in Section~\ref{sec:future}.

\subsection{Formulations of Fast Code Surgery}
\label{sec:intro_fast_surgery}

Most prior works on code surgery, whether on surface codes or generic qLDPC codes, take as standard $O(d)$ time per operation. 
This barrier is due to the aforementioned fact that for generic codes, extracting the stabilizer measurement values reliably under the impact of noise requires $O(d)$ rounds of repeated measurement. 
To overcome this barrier, two forms of fast surgery has been proposed and studied.

The most direct form is where we measure the base code, the deformed code, and then the base code again each for a single round, and the protocol is fault tolerant with respect to a reasonable measure, such as fault distance. 
This form of surgery, which we henceforth call \textit{single-shot surgery}, is possible for certain base codes which admit single-shot state preparation~\protect\footnote{We note that single-shot state preparation is different from, and often stronger than, the property of single-shot error correction~\protect\cite{bombin2015single,campbell2019theory}, which has several variations in the literature and is known to hold on a wider collection of qLDPC codes~\protect\cite{fawzi2018constant,gu2024single-shot,scruby2024high-threshold}.}, such as three-dimensional HGP codes~\cite{hillmann2024single-shot,Tan2025Single}. 
However, for more general codes this style of surgery is no longer fault-tolerant: when we return to the base code after deformation, we need to measure out the ancilla system, which non-deterministically induces a \textit{frame error} on the base code. We briefly explain what a frame error is in footnote~\footnote{To understand what a frame error is, consider the usual state preparation protocol where we prepare $n$ physical qubits in $\ket{+}$ state, and then measure the stabilizers of a CSS code. 
In the case where no errors occur, all $X$ stabilizers will measure to $+1$, while the $Z$ stabilizers will have non-deterministic outcomes. 
In other words, the $Z$ frame experienced a random shift, hence the name ``frame error''.
This frame shift will be further clouded by stochastic errors, which makes it difficult to be recovered and corrected. 
The usual method is therefore to repeatedly measure the stabilizers for $O(d)$ rounds, so that the syndromes are reliable and the frame shift may be recovered. 
Some codes, such as three- and higher-dimensional HGP codes, admit single-shot state preparation, which means there is enough redundancy in the stabilizers such that $O(1)$ rounds of measurement is sufficient for us to extract the frame shift.
This is why single-shot surgery is possible for such codes.}. 
For most base codes which do not admit single-shot state preparation, this frame error cannot be corrected from a single round of syndrome information, and will corrupt further logical operations if we proceed without correction.
This perspective is detailed further in Section~1.2 and 1.3 of Ref.~\cite{cowtan2025fast}.

\begin{figure}[t]
    \centering
    \includegraphics[width=0.8\linewidth]{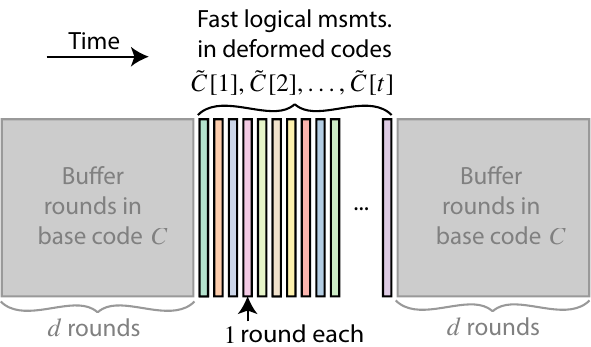}
    \caption{\textbf{Constant-time Surgery.} We abstractly represent the spacetime volume of $t$ surgery operations, where time flows from left to right. 
    The thin slices represent switching between a sequence of deformed codes, $\tilde{C}[1]...\tilde{C}[t]$, each performing a different logical measurement. 
    Importantly, we only need to spend a single round in each deformed code.
    To ensure fault tolerance in the presence of noisy measurements, these $t\geq d$ fast surgery measurements must be preceded and followed by $O(d)$ rounds of syndrome measurements in the base code $C$.}
    \label{fig:amortized}
\end{figure}

To design more broadly applicable fast surgery protocols, we therefore need a different framework.
In this paper, we work with the \textit{fast hypergraph surgery} methods of Ref.~\cite{cowtan2025fast}. 
Instead of asking for one surgery operation to be fault-tolerant within $O(1)$ rounds, Ref.~\cite{cowtan2025fast} asks for $t\ge O(d)$ surgery operations to be fault-tolerant within $O(t)$ rounds. 
More precisely, suppose we have base code $C$ and a sequence of 
deformed codes $\tilde{C}[1], \cdots, \tilde{C}[t]$.
The fast surgery procedure is to start with $O(d)$ rounds of syndrome measurements of the base code, then measure the stabilizers of $\tilde{C}[1], \cdots, \tilde{C}[t]$ in order, each for $O(1)$ rounds, and finally finish with $O(d)$ syndrome measurement rounds of the base code.
In this way, each surgery operation costs a constant number of rounds in amortization~\footnote{In general, one may choose to perform $t<d$ logical operations in between buffering syndrome measurement rounds of the base code; however, this means that the overhead per surgery operation is increased.}.
We henceforth refer to this second form of fast surgery as \textit{constant-time surgery}.
See Figure~\ref{fig:amortized} for an illustration.

Why can such an amortized protocol be fault-tolerant, while each individual operation, strictly speaking, is not? 
We give a brief intuitive argument and refer the reader to Ref.~\cite{cowtan2025fast} for detailed proofs. 
On a high level, performing one surgery operation in a constant number of rounds does not generate sufficient syndrome information to ensure fault-tolerance. 
However, by performing a sequence of $t$ surgery operations (which satisfy certain conditions, see Lemma~\ref{lem:fast-surgery-conditions}) in $O(t)$ time, buffered by $O(d)$ syndrome rounds on the base code, we can generate enough syndrome information to decode all of them together and recover from the sequence of frame errors. 
In technical terms, while each surgery operation introduces new stabilizers and frame errors, a sufficient number of detectors are generated at each step, and we have an $\Omega(d)$ length chain of detectors to decode for fault-tolerance. 
This idea of using detectors across multiple logical operations to ensure fault tolerance in amortization is akin to algorithmic fault-tolerance~\cite{zhou2024algorithmic} for transversal gates.

We note that in a common computational model, sequential Pauli-based computation \cite{bravyi2016trading,litinski2019game}, constant-time surgery has the potential to reduce the overall time overhead by a factor of $O(d)$, because logical measurements are the only operations in the computation, provided a supply of high-quality magic states at constant rate. 
If we further include unitary logical Clifford operations (e.g.~transversal gates), we would expect the fault-tolerance guarantee for constant-time surgery to be preserved as the detector chains stay intact. 
However, in a setting where logical non-Cliffords are performed unitarily and frequently, this would break up the detector chains, complicating the use of constant-time surgery.

We note that the recent work of Ref.~\cite{baspin2025fast} studied fast surgery as well, in the following sense:
they perform a single surgery operation in $O(1)$ time, preceded and followed by perfect syndrome measurements on the base code. 
They derived a set of conditions (Theorem~5 of Ref.~\cite{baspin2025fast}) to guarantee that the above procedure is fault tolerant, and presented constructions which use $\tilde{O}(nkd)$ ancilla qubits per operation.
As discussed earlier, obtaining the effect of perfect syndrome information on generic codes takes $O(d)$ rounds of noisy syndrome measurement~\footnote{We note that there are a few ways to obtain reliable syndrome data on generic codes. The most straightforward way is to measure the stabilizers for $O(d)$ rounds. One could also use Steane-style measurement~\protect\cite{steane1997active,knill2005quantum} which requires a well-prepared ancilla code state.
In this case, the code state takes $O(d)$ rounds to prepare, but the syndrome measurement on target code takes $O(1)$ time.
Finally, one could also readout all the physical qubits of the code, which unfortunately destroy the code state.}. 
Therefore, the surgery gadgets in Ref.~\cite{baspin2025fast} are fast in the sense of single-shot surgery when the base code admits single-shot state-preparation, and cost $O(d)$ time per operation otherwise.

\subsection{Overview of Construction}\label{sec:summary_construction}

We now describe, at a high-level, the constant-time surgery gadgets we construct in this paper.  
Let the code we would like to perform surgery on be the hypergraph product of two classical codes, $C$ and $D$. 
We will attach a surgery gadget to the classical code $C$, with known techniques~\cite{williamson2024low-overhead}, creating an augmented classical code that we refer to as $\cone(g)$ that measures a particular classical codeword (this notation comes from the fact that the augmented classical code is a mapping cone with respect to a chain map $g$, see Section~\ref{sec:prelims}). 
Then, the final, deformed code is constructed by taking hypergraph product of $\cone(g)$ with $D$.
This deformed code measures, in parallel, all \textit{quantum} logical operators whose canonical representatives on the product code correspond with the measured classical codeword (see Section~\ref{sec:prelims}).
By taking a sequence of different surgery gadgets on the classical code, giving $\cone(g_1), \cdots, \cone(g_t)$, we obtain a sequence of deformed codes.
Importantly, we will show that this sequence of deformed codes satisfies all the properties required for fault-tolerant constant-time surgery in Ref.~\cite{cowtan2025fast}.

\begin{figure*}[t]
    \centering
    \includegraphics[width=0.8\linewidth]{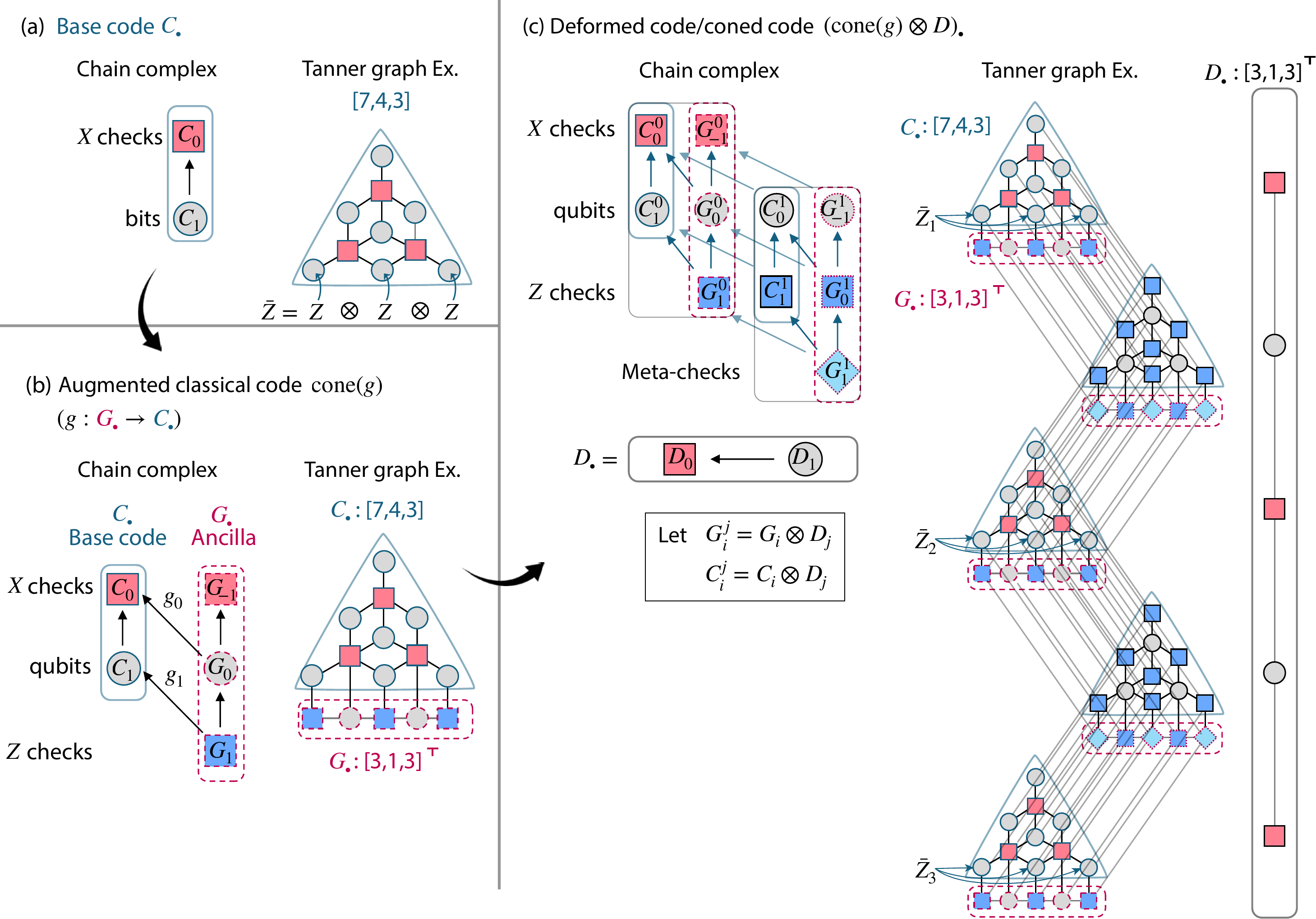}
    \caption{
    \textbf{(a)} A classical code, $C_\bullet=C_1\rightarrow C_0$ used in the base hypergraph product code $Q$, $(C\otimes D)_\bullet$. 
    A concrete example, the $[7,4,3]$ Hamming code, is shown with a highlighted logical operator $\bar{Z}$. 
    \textbf{(b)} The augmented classical code, $\cone(g)$, which is the classical code $C_\bullet$ attached to an ancilla complex $G_\bullet: G_1\rightarrow G_0\rightarrow G_{-1}$  through a chain map, $g_i:G_i\rightarrow C_i$.
    The Tanner graph of the augmented $[7,4,3]$ code is shown.
    Notably, the product of all $Z$ checks of $G_\bullet$ is equal to the previously highlighted $\bar{Z}$ logical.
    \textbf{(c)} The final deformed code or coned code, $(\cone(g)\otimes D)_\bullet$, that we switch into to perform fast surgery on $(C\otimes D)_\bullet$. $(\cone(g)\otimes D)_\bullet$ is the tensor product of the augmented code in~\ref{fig:toric-code}(b) with $D_\bullet$.
    This deformed code contains meta-checks (in $G_{1}^1$ of the chain complex), which check for measurement errors of $Z$ checks.
    We show the Tanner graph of a concrete example of a deformed code by taking the hypergraph product of the augmented code example in~\ref{fig:toric-code}(b) with the dual of the $[3,1,3]$ repetition code.
    This coned code measures the highlighted logical representatives, $\bar{Z}_1, \bar{Z_2}$ and $\bar{Z}_3$, in parallel.
    }    
    \label{fig:summary}
\end{figure*}

We show an explicit example of one such deformed code in Fig.~\ref{fig:summary}.
For concreteness, let $C$ be the $[7,4,3]$ Hamming code (in Fig.~\ref{fig:summary}~(a)) and let $D$ be the dual of the linear $3$-bit repetition code.
The hypergraph product $Q=C\otimes D$ of these codes is a $[[27,4,3]]$ quantum code, and our goal is to measure a logical operator of code $Q$.
In Fig.~\ref{fig:summary}~(a), we represent code $C$ as a chain complex (left) $C_\bullet:C_1\rightarrow C_0$ where $C_1$ is a linear space spanned by all bits, and $C_0$ is the space spanned by all $X$ checks.
The Tanner graph of $[7,4,3]$ (right) and the support of a $\bar{Z}$ logical is also presented.
Then, in Fig.~\ref{fig:summary}~(b), we apply code surgery to measure the $\bar{Z}$ logical on the bottom three qubits of the $[7,4,3]$ Tanner graph.
This entails the introduction of a new ancilla code, $G_\bullet = G_1\rightarrow G_0\rightarrow G_{-1}$ (dotted red box) with $Z$ checks ($G_1$), qubits ($G_0$), and $X$ checks ($G_{-1}$).
The connection between this ancilla code and the original code is formally stated as a chain map $g_i:G_i\rightarrow C_i$.
We refer to this deformed classical code as $\cone(g)$, and refer the reader to Section~\ref{sec:prelims} for more background on coning maps. 
In this example, $G_\bullet$ is the dual of the $3$-bit repetition code. 
We see that a product of all of the ancilla $Z$ checks equal the $\bar{Z}$ logical operator representative on $C_\bullet$, ensuring that the measurement of the ancilla stabilizers will measure $\bar{Z}$ as intended.
We remark that for this particular choice of $G_\bullet$, $G_{-1}$ is empty, however there may be $X$ checks in $G_\bullet$ in general.
Finally, we arrive at the full deformed code in Fig.~\ref{fig:summary}~(c) by taking the product of $\cone(g)$ with $D_\bullet$, the other classical code in the base hypergraph product code.
We refer to this deformed code as the coned code, $(\cone(g)\otimes D)_\bullet$.
The full chain complex of the coned code is drawn on the left of Fig.~\ref{fig:summary}~(c), where we use $G^j_i$ to denote the space $G_i\otimes D_j$, and $C^j_i$ to denote the space $C_i\otimes D_j$.\

Notably, the merged code has meta-checks in $G_{1}^1$ which check for measurement errors on ancilla $Z$ checks in $G_1^0$ and $G_0^1$ in addition to $Z$ checks of the original code, $C_1^1$.
We discuss the importance of these meta-checks for fast surgery in Section~\ref{sec:fast-surgery}.
Furthermore, for our concrete example, $D$ is the dual of the three-bit linear repetition code.
The Tanner graph of our $(\cone(g)\otimes D)_\bullet$ example is on the right. We also highlight three logical operators $\bar{Z}_1,\ \bar{Z}_2,\ \bar{Z}_3$ of the hypergraph product code $Q$, which are in fact all logically equivalent representatives of a single logical operator.
By inspection, we can see that we achieve the aforementioned goal of measuring a logical of $Q$.
Namely, each representative $\bar{Z}_i$ is being measured in the construction $(\cone(g)\otimes D)_\bullet$ because the representative is equal to a product of ancilla $Z$ checks, the three dark blue squares below it. 
Intuitively, the redundancy of measuring different representatives of the same logical operator ensures fault-tolerance against any single logical measurement being flipped by a measurement error, which is the same fault that would in normal non-constant-time surgery be caught by repeating measurements.
For further technical discussion on the importance of measuring multiple representatives for fast surgery, we refer the reader to Appendix~E of Ref.~\cite{cowtan2025fast}.
We emphasize that this example is one deformed code, which is only fault-tolerant and constant-time when used in quick succession with other deformed codes corresponding to choosing different classical code surgery $\cone(g_i)$.

\subsection{Prior Works}
\label{sec:previous-works}

\paragraph{Overheads of existing fast surgery gadgets.}

For three- and higher-dimensional HGP codes, which admit single-shot state preparation in one or both ($X$ and $Z$) bases, prior works~\cite{hillmann2024single-shot,Tan2025Single} observed that constant space overhead single-shot surgery can be performed. 
For other codes which admit single-shot state preparation, Ref.~\cite{baspin2025fast} constructed single-shot surgery gadgets which can measure an arbitrary Pauli operator from the code, at a daunting $\tilde{O}(nkd)$ qubit overhead. 

Ref.~\cite{cowtan2025fast} formulated constant-time surgery, and proposed \textit{block reading} as an example.   
Block reading is applicable to any CSS code, and incurs a $O(n)$ qubit overhead. 
However, it has limited flexibility in terms of its logical action. 
Block reading measures all logical qubits in one or multiple blocks of the same CSS codes in parallel, but cannot address individual or subsets of qubits without strong assumptions on the base codes. 
For instance, given two code blocks of a CSS code, block reading can be applied to measure the $Z\otimes Z$ observable on all indexed pairs of logical qubits simultaneously. 
Given $b$ code blocks and a set of $t$ commuting $b$-qubit Pauli observables, we can use block reading to measure these Pauli observables on all $k$ indexed tuples of $b$ logical qubits in $O(t)$ time, in the model of constant-time surgery. 
We note that this logical action can also be achieved using transversal CNOT. 
From this perspective block reading establishes an interesting connection between surgery and transversal gates, as well as between constant-time surgery and algorithmic fault-tolerance~\cite{zhou2024algorithmic}. 
We further note that the logical action of block reading can be used to implement CSS code concatenation, similar to those considered in the single-shot code-switching techniques of Refs.~\cite{Tan2025Single,golowich2025constant,xu2025batched}. 
These connections would be interesting to explore in future works.

\paragraph{Logical operations on HGP codes.}

HGP codes are extensively studied in prior literature due to their desirable structural properties~\cite{eczoo_hypergraph_product}. 
We limit our discussions to closely related works, including Refs.~\cite{xu2024fast,zheng2025high,Tan2025Single,golowich2025constant,xu2025batched}.

To discuss Ref.~\cite{xu2024fast}, we briefly describe the technique of \textit{homomorphic measurement} introduced in Ref.~\cite{huang2023homomorphic}. 
Similar to surgery, homomorphic measurement is a way to measure logical Pauli observable from a base code. 
Instead of a deformed code, homomorphic measurement uses an ancillary code block of sufficient distance, prepares this code block in a stabilizer state, and applies a CNOT circuit (often called \textit{homomorphic CNOT}) between the ancilla and base code blocks to entangle the logical qubits. 
The physical qubits of the ancilla block is then measured out, and the desired logical measurement results on the base code can be inferred from the physical measurement outcomes on the ancilla block.

The efficacy of homomorphic measurement varies significantly depending on the structure of the base code. 
On homological product codes, Ref.~\cite{xu2024fast} designed ancilla code blocks of size $O(n)$ (which are homological product codes themselves) that can be used to perform Pauli-product measurements (PPM) in parallel. 
By using different ancilla codes and constant-depth unitary gates~\cite{quintavalle2023partitioning,berthusen2025automorphism}, Ref.~\cite{xu2024fast} showed that the full logical Clifford group on HGP codes can be implemented. 
Note that as discussed earlier, preparing the ancilla code block in stabilizer states takes $O(d)$ time generically. 
Therefore, for 2D HGP codes, their gadgets cost $O(d)$ time per operation. 
In comparison, our surgery gadgets cost $O(1)$ time per operation. 

Further insights and observations can be derived from comparing our constructions with those in Ref.~\cite{xu2024fast}. 
Notably, Ref.~\cite{cowtan2025fast} showed a circuit-equivalence between CSS surgery and homomorphic measurements via ZX-calculus.
Moreover, our surgery ancilla systems are HGP codes as well, drawing parallel to the homological product ancilla codes used by Ref.~\cite{xu2024fast}.
These similarities and connections invite natural questions: what are the significant distinctions between the two constructions, and why is one faster? 
We answer these questions at the end of Section~\ref{sec:logical-action}, after we introduce the necessary technical tools.

More recently, Ref.~\cite{zheng2025high} constructed surgery gadgets for 2D HGP codes that can measure up to $O(k)$ operators in parallel (in $O(d)$ time), and proposed to utilize classical code homomorphisms to help the gadgets measure a more flexible set of operators compared to what's possible with homomorphic measurements. 
The differences between their construction and ours are best seen by comparing chain complexes, see equation~(D19) in Ref.~\cite{zheng2025high} and Fig.~\ref{fig:explicit-toric-chain-map} of this paper. 
At a high-level, the main difference is that we construct a $4$-term chain complex for $\cA$ instead of a $3$-term chain complex.
Here, the extra term in $\cA$ is the meta-checks we use to check for measurement errors in place of repeated rounds in the deformed code.

For certain families of three- and higher-dimensional HGP codes, Refs.~\cite{Tan2025Single,golowich2025constant} showed how to switch, in constant time, between lower-dimensional (including two-dimensional) ``slices'' of the codes and higher-dimensional ``slices''.
Utilizing these single-shot code-switching and single-shot state-preparation methods, they showed how to implement various gates on encoded logical information in constant time. 
Independently, Ref.~\cite{xu2025batched} showed that single-shot code-switching is possible for any two families of quantum CSS codes by taking their homological product. 
This primitive enables them to implement ``batched'' Clifford circuits in constant time. Here ``batched'' means that the same $k$-qubit circuit is being implemented on a large enough number of identical code blocks, each encoding $k$ logical qubits.
It is worth noting that these single-shot code-switching techniques~\cite{Tan2025Single,golowich2025constant,xu2025batched} all require codes that are at a larger scale than 2D HGP codes, and are likely to have a lower rate.

\section{Preliminaries}
\label{sec:prelims}

\subsection{Chain Complexes and Quantum Codes}

We first introduce standard notations and definitions on chain complexes and quantum codes, which will be used throughout this paper.

A chain complex $\cC$ over $\bF_2$ is a sequence of based vector spaces equipped with linear maps, which we call boundary maps, between successive spaces. 
\begin{equation}
\begin{tikzcd}
\cC = \cdots \arrow[r,"\partial_{n+2}"] & C_{n+1} \arrow[r,"\partial_{n+1}"] & C_n \arrow[r,"\partial_n"] & C_{n-1} \arrow[r,"\partial_{n-1}"] & \cdots
\end{tikzcd}
\end{equation}
The boundary maps satisfy the condition that $\partial_{i-1}\circ \partial_{i} = 0$, or equivalently, $\im(\partial_{i})\subseteq \ker(\partial_{i-1})$. 
We consider bounded chain complexes, which means a finite number of spaces $C_i$ are non-trivial. 
The co-chain complex $\coC$ is constructed by taking the transpose of the boundary operators as co-boundary maps $\delta^{i-1} = \partial_i^\top$. 
\begin{equation}
\begin{tikzcd}
\coC = \cdots & \arrow[l,"\delta^{n+1}"']  C^{n+1} & \arrow[l,"\delta^{n}"']  C^n & \arrow[l,"\delta^{n-1}"']  C^{n-1} &  \arrow[l,"\delta^{n-2}"'] \cdots
\end{tikzcd}
\end{equation}
Here we identify $C^i\cong C_{i}$ for all indices $i\in \mathbb{Z}$. 
The $i$-homology and $i$-co-homology spaces are defined as 
\begin{align}
    H_i(\cC) &= \ker(\partial_i)/\im(\partial_{i+1}), \\
    H^i(\coC) &= \ker(\delta^i)/\im(\delta^{i-1}).
\end{align}
The $i$-systolic distance $d_i(\cC)$ and the $i$-cosystolic distance $d^i(\coC)$ are
\begin{align}
    d_i(\cC) &=\min_{x \in H_i(\cC)}|x|, \\
    d^i(\coC) &=\min_{x\in H^i(\coC)}|x|.
    \label{eq:systolic-cosystolic-distance}
\end{align}
We often abbreviate the notations as $d_i(C)$ and $d^i(C)$. 
In case the (co-)homology spaces are empty, we follow the convention in Ref.~\cite{zeng2020minimal} and set the distances to in $fty$. 
This convention will be helpful for our later distance analysis.

A quantum CSS codes can be represented by a length-2 chain complex,
\begin{equation}
\begin{tikzcd}
\cC = C_{2} \arrow[r,"\partial_{2}"] & C_1 \arrow[r,"\partial_1"] & C_{0},
\end{tikzcd}
\label{eq:css-chain}
\end{equation}
where the 2-boundary is assigned to the transpose of the $Z$-parity check matrix, $\partial_2=H_Z^\top $, and the 1-boundary is assigned to the $X$-parity check matrix, $\partial_1 =H_X$. 
The condition $\partial_1\circ\partial_2=0$ is exactly the requirement that stabilizers commute, $H_XH_Z^\top =0$.
$Z$-logical operators are in the 1-homology, $H_1(\cC) = \ker(H_X)/\im(H_Z^\top )$, because they are undetected by parity checks in $H_X$, and not in the $Z$-stabilizer group.
Similarly, $X$-logical operators are in the 1-cohomology, $H^1(\coC)=\ker(H_Z)/\im(H_X^\top )$. 
The $X$- and $Z$-distances of the code are the $1$-cosystolic and systolic distances, respectively.

For this work, we often consider multiple (co-)chain complexes at the same time. For this reason, from now on we will label the (co-)boundary maps with subscripts indicating their complexes, except when the complex is evident from context. 

\subsection{CSS Code Surgery as Mapping Cones}

In this section we briefly review the methods of code surgery, which enable us to design fault-tolerant schemes to measure Pauli logical operators of quantum stabilizer codes. 
Logical measurement is a powerful primitive for fault-tolerant quantum computation (FTQC), as described by the model of Pauli-based computation (PBC)~\cite{bravyi2016trading}. 
Specifically, Clifford gates can be implemented using Pauli measurements and feed-forward Pauli corrections, and non-Clifford gates can be implemented similarly with the help of magic states. For more details see, for instance, Figure~11 of Ref.~\cite{he2025extractors}.
Following recent developments~\cite{cohen2022low-overhead,cowtan2024ssip,cross2024improved,williamson2024low-overhead,Ide_2025,swaroop2024universal}, code surgery is now a practically promising and generally applicable approach for FTQC on qLDPC codes~\cite{he2025extractors,yoder2025tour}.

For our purposes, we restrict to the case of CSS codes and largely follow the presentation in Ref.~\cite{Ide_2025}.
Suppose that we would like to measure a $Z$-logical operator of a quantum code, represented as a chain complex $\cC$. At a high-level, we attach an ancilla system to $\cC$ and measure $Z$-checks of the ancilla system to extract the value of the logical information. 
Formally, code surgery introduces an ancilla system 
\begin{equation}
\begin{tikzcd}
\cA = A_{1} \arrow[r,"\partial_{A, 1}"] & A_0 \arrow[r,"\partial_{A, 0}"] & A_{-1},
\end{tikzcd}
\label{eq:anc-chain}
\end{equation}
and a chain map $f_\bullet:A_\bullet\rightarrow \cC$. 
A chain map is a homomorphism between two chain complexes $f_i:A_i\rightarrow C_i$.
The commuting diagram is as follows.
\begin{equation}
\begin{tikzcd}[column sep=large,row sep=large]
|[alias=0]|  0 \arrow[r, "0"] \arrow[d, "0"'] &
|[alias=A1]| A_1 \arrow[r, "\partial_{A,1}"] \arrow[d, "f_1"'] &
|[alias=A0]| A_0 \arrow[r, "\partial_{A,0}"] \arrow[d, "f_0"'] &
|[alias=Am1]| A_{-1} \arrow[d, "0"'] \\
|[alias=C2] |C_2 \arrow[r, "\partial_{C,2}"'] &
|[alias=C1]| C_1 \arrow[r, "\partial_{C,1}"'] &
|[alias=C0]| C_0 \arrow[r, "0"] &
0
\end{tikzcd}
\label{eq:chain-map}
\end{equation}
In particular, $f_0 \partial_{A,1}= \partial_{C,1}f_1$.
We take the mapping cone of $f$, denoted $\cone(f)$, which is obtained by left-shifting the ancillary complex $\cA$ by one degree.
\begin{equation}
\cone(f) = 
\begin{tikzcd}[column sep=large,row sep=large]
|[alias=A1]| A_1 \arrow[r, "\partial_{A,1}"] \arrow[rd, "f_1"'] &
|[alias=A0]| A_0 \arrow[r, "\partial_{A,0}"] \arrow[rd, "f_0"'] &
|[alias=Am1]| A_{-1} \\
|[alias=C2] |C_2 \arrow[r, "\partial_{C,2}"'] &
|[alias=C1]| C_1 \arrow[r, "\partial_{C,1}"'] &
|[alias=C0]| C_0
\end{tikzcd}
\label{eq:cone-code-surgery}
\end{equation}
The fact that $f$ is a homomorphism, or equivalently the fact that the diagram in equation~\eqref{eq:chain-map} commutes, implies that $\cone(f)$ is a valid chain complex. 

One can view~\eqref{eq:cone-code-surgery} as a CSS code itself, with $Z$-checks, qubits, and $X$-checks living the vector spaces $A_1\oplus C_2$, $A_0\oplus C_1$, and $A_{-1}\oplus C_0$ respectively. 
In the context of surgery, this code is called the \textit{merged code} or \textit{deformed code} with parity check matrices
\begin{align}
    \tilde{H}_Z^\top =\begin{pNiceMatrix}[first-row, first-col]
        & A_1 & C_2 \\
        A_0 & \partial_{A,1} & 0 \\
        C_1 & f_1 & \partial_{C,2}
    \end{pNiceMatrix}\label{eq:graph-surg-hz}\\
    \tilde{H}_X=\begin{pNiceMatrix}[first-row, first-col]
        & A_0 & C_1 \\
        A_{-1} & \partial_{A,0} & 0\\
        C_0 & f_0 & \partial_{C,1}
    \end{pNiceMatrix}\label{eq:graph-surg-hx}
\end{align}

How does the deformed code help us perform logical measurements? As detailed in prior works~\cite{cross2024improved,williamson2024low-overhead,Ide_2025}, the idea is to perform code-switching between the base memory code and the deformed code.
We start in the base code with reliable syndrome information, then measure the new stabilizers of the deformed code for $O(d)$ rounds, and finally return to the base code and measure its stabilizers for $O(d)$ rounds.
In this process, all $Z$-operators of the base code corresponding to vectors in 
\begin{align}
    f_1(\ker(\partial_{A,1})) = \{f_1(v): v\in \ker(\partial_{A,1})\} \subseteq \ker(\partial_{C,1}).
    \label{eq:measured-operators}
\end{align}
are measured by the ancilla system.
To see this, consider the operator $Z^{f_1(v)}$, which are Pauli $Z$s on qubits in $f_1(v)$.
Note that such an operator must be a $Z$-stabilizer or logical operator of $\cC$, and is the product of the set of deformed $Z$-checks in $A_1$ corresponding to non-zero indices in $v$. 
Therefore, its measurement result can be inferred from the measurement values of the deformed checks, which are reliable since we repeatedly measured them for $O(d)$ rounds. 
For a more detailed derivation, see, for instance, Lemma~1 of Ref.~\cite{Ide_2025}.
Note that the protocol for fast quantum surgery is different, as we will detail in Section~\ref{sec:fast-surgery}.

The hard work, therefore, lies in constructing appropriate ancilla complexes $\cA$ such that the code-switching protocol performs the correct measurements fault-tolerantly, with low-overhead. 
The leading method to construct a surgery system which measures a chosen logical operator from an arbitrary CSS code is proposed in Ref.~\cite{williamson2024low-overhead} as \textit{gauging measurements} and independently in Ref.~\cite{Ide_2025} as \textit{homological measurements}. 
We summarize their main results as follows, in the language of chain complexes. 

\begin{lemma}[Main results from~\cite{williamson2024low-overhead,Ide_2025}]\label{lem:graph-surgery}
    Consider a three-term complex $\begin{tikzcd}
\cC = C_{2} \arrow[r,"\partial_{C,2}"] & C_1 \arrow[r,"\partial_{C,1}"] & C_{0},
\end{tikzcd}$ and an element $L\in H_1(C_\bullet)$. 
We can construct an ancillary complex $\cA$, a chain map $f:\cA\rightarrow \cC$, and therefore the mapping cone $\cone(f)$ as in equation~\eqref{eq:anc-chain},~\eqref{eq:chain-map} and~\eqref{eq:cone-code-surgery}, such that the following conditions hold.
\begin{enumerate}
    \item The map $\partial_{A, 1}$ has boundary Cheeger constant at least 1 (see Definition~\ref{def:Cheeger}). \label{surgery-cond:expansion}
    \item $\dim(H_0(\cA)) = 0$, $\dim(H_{-1}(\cA)) = 0$. \label{surgery-cond:no-gauge-qubits}
    \item $\dim(H_1(\cA)) = 1$, and the only element $v$ in $H_1(\cA) = \ker(\partial_{A, 1})$ is the all ones vector on $A_1$. It satisfies $f_1(v) = L$. 
    \label{surgery-cond:ancilla-homology}
    \item There is a constant $w$ such that the boundary maps of $A$ and the chain maps are $w$-sparse; i.e., the row and column weights of $\partial_{A, 1}$, $\partial_{A, 0}$, and $f_0$ are all upper bounded by $w$. Moreover, $f_1$ is $1$-sparse.
    \label{surgery-cond:sparsity}
    \item $\dim(A_1) + \dim(A_0) + \dim(A_{-1})\le O(|L|\log^3(|L|))$.\label{surgery-cond:size-bound}
\end{enumerate}
Using these conditions, it was shown that $d_1(\cone(f))\ge d_1(C)$ and $d^1(\cone(f))\ge d^1(C)$. 
\end{lemma}

\begin{definition}[Boundary Cheeger Constant]\label{def:Cheeger}
    Consider a matrix $M: \bF_2^n\rightarrow \bF_2^m$ such that the all ones vector $v$ is in the kernel of $M$. 
    Its boundary Cheeger constant $\beta(M)$ is the maximum value such that for all $v\in \bF_2^n$, we have
    \begin{align}
        |Mv| \ge \beta(M)\cdot \min(|v|, n - |v|).
    \end{align}
\end{definition}

For clarity, let us explain the above conditions in context of surgery on CSS codes. 
First of all, condition~\ref{surgery-cond:expansion} is key to proving the conclusions of the lemma, which states that the deformed code, $\cone(f)$, has at least the same $X$ and $Z$-distances as the base code $\cC$. 
Condition~\ref{surgery-cond:no-gauge-qubits} ensures that the ancilla system $\cA$ has no logical qubits by itself, and therefore attaching it to $\cC$ does not introduce new logical degrees of freedom, which are often called {gauge logical qubits}. 
Condition~\ref{surgery-cond:ancilla-homology} states that the logical operator $L$ becomes a product of stabilizers in the deformed code, as desired for logical measurements (equation~\eqref{eq:measured-operators}). 
Condition~\ref{surgery-cond:sparsity} bounds the stabilizer check weights in the deformed code. 
Finally, condition~\ref{surgery-cond:size-bound} bounds the space overhead, which is the total number of checks and qubits, introduced by the ancilla system. 

Given these conditions, starting in a base CSS code with distance $d$, Ref.~\cite{williamson2024low-overhead} showed that the aforementioned code-switching protocol has \textit{fault distance} $d$. 
Formally, they considered two type of physical noise: Pauli errors on physical qubits and flip errors on stabilizer measurements. 
Fault distance captures the minimum number of such errors needed in the protocol to corrupt the logical information without being detected (by detectors in the protocol). 
This measure of fault-tolerance is standard in the study of code surgery and serves as an indicator for protocol performance under circuit level noise. For the rest of this paper, we will consider the fault distance under this noise model. 

We remark that the gauging measurement technique from Ref.~\cite{williamson2024low-overhead} is more general than as stated in Lemma~\ref{lem:graph-surgery}, as it can be used to measure mixed-type operators (i.e., products of $X, Y, Z$ operators on different qubits) on non-CSS stabilizer codes. 
This is critical in application such as the extractor architectures of Ref.~\cite{he2025extractors}, which when given a stabilizer code builds an \textit{extractor} that is capable of extracting any desired logical Pauli operator from the base code.

We further note that while Lemma~\ref{lem:graph-surgery} is restricted to measuring a single Pauli operator (which could be supported on many logical qubits), prior works have also studied simultaneous measurements (or equivalently parallel measurements) of multiple commuting Pauli operators using surgery~\cite{zhang2025time-efficient,cowtan2025parallel,zheng2025high} or homomorphic measurement~\cite{huang2023homomorphic,xu2024fast}.
The gadgets we build in this paper also measure multiple operators in parallel, see Section~\ref{sec:logical-action}. 

\subsection{Fast Code Surgery}\label{sec:fast-surgery}

In code surgery protocols, the logical measurement is inferred from the measurement of a product of $Z$-stabilizers in the deformed code.
In a setting where measurement errors can occur, it is natural to use $O(d)$ rounds of measurement to ensure the stabilizer measurements are reliably extracted so that the logical measurement has a fault distance $d$. 
However, there is another approach through which stabilizer values can be reliably extracted without repeated rounds of measurement, that is, using \textit{meta-checks}. 
Specifically, some quantum codes have redundancy in their stabilizer generators such that the product of certain stabilizer outcomes must necessarily be $+1$ in the absence of measurement errors. 
These products of stabilizers are called meta-checks, which can be interpreted as classical checks on syndromes. 
Meta-checks have been commonly used to reduce the time overhead of syndrome extraction in memory. 
Ref.~\cite{cowtan2025fast} employed them to reduce the time overhead of surgery. 
We briefly review their framework here.

Formally, suppose a CSS code has $m_z$ different $Z$-stabilizer checks.
We denote meta-checks on $Z$-stabilizers as $M_Z\in \bF_2^{r_z\times m_z}$, where $M_ZH_Z=0$. 
The number of independent meta-checks is $r_z$. 
Meta-checks can be seen as an additional term in the chain complex that describes the quantum code. 
Suggestively, we will add another term to the ancilla chain complex to represent meta-checks on the $Z$-stabilizers.
\begin{equation}
\begin{tikzcd}[column sep=large]
\cA = A_{2} \arrow[r,"\partial_{A, 2}=M_Z^\top "] & A_{1} \arrow[r,"\partial_{A, 1}=H_Z^\top"] & A_0 \arrow[r,"\partial_{A, 0}=H_X"] & A_{-1},
\end{tikzcd}
\label{eq:meta-cochain}
\end{equation}
We define the meta-check distance of $\cA$ to be the 1-cosystolic distance,
\begin{align}
    d^1(\coA) &= \min_{x\in H^1(\coA)}|x|, \\
    H^1(\coA) &= \frac{\ker(\delta_{A}^2)}{\im (\delta_A^1)} = \frac{\ker(M_Z)}{\im (H_Z)}.
    \label{eq:1-cosystolic-distance}
\end{align}
As before, we can construct a chain map $f:\cA\rightarrow \cC$.
\begin{equation}
\begin{tikzcd}[column sep=large,row sep=large]
|[alias=A2]| A_2 \arrow[r, "\partial_{A,2}"] \arrow[d, "f_2"'] &
|[alias=A1]| A_1 \arrow[r, "\partial_{A,1}"] \arrow[d, "f_1"'] &
|[alias=A0]| A_0 \arrow[r, "\partial_{A,0}"] \arrow[d, "f_0"'] &
|[alias=Am1]| A_{-1} \arrow[d, "0"'] \\
|[alias=C2] |C_2 \arrow[r, "\partial_{C,2}"'] &
|[alias=C1]| C_1 \arrow[r, "\partial_{C,1}"'] &
|[alias=C0]| C_0 \arrow[r, "0"] &
0
\end{tikzcd}
\label{eq:4-term-chain-map}
\end{equation}
The mapping cone $\cone(f)$ is now a $4$-term chain complex.
\begin{equation}
\begin{tikzcd}[column sep=large,row sep=large, /tikz/execute at end picture={
    \node (huge) [fit=(A2),label=above:{Metachecks}] {};
    \node (huge) [fit=(A1) (C2),label=above:{Z checks}] {};
    \node (huge) [fit=(A0) (C1),label=above:{qubits}] {};
    \node (huge) [fit=(Am1) (C0),label=above:{X checks}] {};
  }]
|[alias=A2]| A_2 \arrow[r, "\partial_{A,2}"] \arrow[rd, "f_2"'] &
|[alias=A1]| A_1 \arrow[r, "\partial_{A,1}"] \arrow[rd, "f_1"'] &
|[alias=A0]| A_0 \arrow[r, "\partial_{A,0}"] \arrow[rd, "f_0"'] &
|[alias=Am1]| A_{-1} \\
&
|[alias=C2] |C_2 \arrow[r, "\partial_{C,2}"'] &
|[alias=C1]| C_1 \arrow[r, "\partial_{C,1}"'] &
|[alias=C0]| C_0
\end{tikzcd}
\label{eq:fast-surgery}
\end{equation}
This new deformed code~\eqref{eq:fast-surgery} has $Z$-type metachecks, $Z$-checks, qubits, and $X$-checks, assigned to the vector spaces $A_2$, $A_1\oplus C_2$, $A_0\oplus C_1$ and $A_{-1}\oplus C_0$.
The parity check matrices are
\begin{align}
    M_Z^\top  
    &= \begin{pNiceMatrix}[first-row,first-col]
        & A_2\\
        A_1 & \partial_{A,2} \\
        C_2 & f_2 
    \end{pNiceMatrix},\\
    \tilde{H}_Z^\top  
    &= \begin{pNiceMatrix}[first-row,first-col]
        & A_1 & C_2\\
        A_0 & \partial_{A,1} & 0\\
        C_1 & f_1  & \partial_{C,2}
    \end{pNiceMatrix},\\
    \tilde{H}_X 
    &= \begin{pNiceMatrix}[first-row,first-col]
        & A_0 & C_1\\
        A_{-1} & \partial_{A,0} & 0\\
        C_0 & f_0 & \partial_{C,1}
    \end{pNiceMatrix}.
\end{align}

Let us now recall the constant-time surgery formulation of Ref.~\cite{cowtan2025fast} discussed in Section~\ref{sec:intro_fast_surgery} and Figure~\ref{fig:amortized}. 
Instead of performing a single surgery operation in $O(1)$ time, we consider performing a sequence of surgery operations, specified by ancillary complexes $\cAi{i}$ and chain maps $f[i]$, in $O(d)$ time. 
For this protocol to be fault-tolerant, the surgery operations must satisfy certain conditions. 
We first make the following definition.

\begin{definition}[Compacted Code]
    \label{def:compacted-code}
    Consider a quantum CSS code $\cC$ with a sequence of ancilla complexes $\cAi{1}, \cdots, \cAi{t}$ and chain maps $f[1]: \cAi{1}\rightarrow \cC$, \dots, $f[t]: \cAi{t}\rightarrow \cC$. 
    Consider the trivial extension of these maps $f'[i]: \oplus_{i\in [t]} \cAi{i}\rightarrow \cC$, whose action is the same as $f[i]$ on $\cAi{i}$ and trivial on the rest of the spaces.  
    The \textit{compacted code} $\cc_\bullet$ is defined by taking the mapping cone of all the chain maps together.
    \begin{align}
        \cc_\bullet = \cone\left( \sum_{i\in [t]} f'[i] \right)
    \end{align}
\end{definition}
\noindent In more descriptive terms, the compacted code is obtained by attaching  all the ancilla systems to the code $\cC$ at once. 
Evidently, $\cc_\bullet$ is likely to have low rate and high stabilizer check weight. 
However, the compacted code is a proof construct that will not be built explicitly in the surgery protocol.
The main result of Ref.~\cite{cowtan2025fast} can then be stated as follows.

\begin{lemma}[Fast hypergraph surgery, Theorem 5.3 in Ref.~\cite{cowtan2025fast}]
Let $\cC$ be a CSS code of distance $d$, and consider a sequence of surgery operations defined by complexes $\cAi{1}, \cdots, \cAi{t}$ and chain maps $f[1]: \cAi{1}\rightarrow \cC$, \dots, $f[t]: \cAi{t}\rightarrow \cC$. Then the surgery protocol described in Figure~\ref{fig:amortized} has fault distance $d$ if the following conditions hold:
\begin{enumerate}
    \item The {compacted code} $\cc_\bullet$ has 1-systolic and 1-cosystolic distances at least $d$. \label{fast-cond:compacted-distance} 
    \item \label{fast-cond:metacheck} Each complex $\cAi{1}$ has 1-cosystolic distance at least $d$.
    \item 
    For all $i$, for standard basis vectors $u,v\in A[i]_1$, $f[i]_1(u)$ and $f[i]_1(v)$ are disjoint, i.e., supported on disjoint indices in $C_1$.
    \label{fast-cond:non-overlap}
\end{enumerate}
\label{lem:fast-surgery-conditions}
\end{lemma}
\noindent In this lemma, condition~\ref{fast-cond:compacted-distance} ensures that the protocol is protected from space-like errors, i.e., Pauli errors on data qubits.
Condition~\ref{fast-cond:metacheck} ensures that each surgery operation has sufficient meta-check distance. This distance protects the protocol from time-like errors, i.e., flip errors on stabilizer measurements. 
Finally, condition~\ref{fast-cond:non-overlap} is used to control error propagation in the protocol.
We refer readers to Ref.~\cite{cowtan2025fast} for detailed proof of this lemma and example constructions, including block reading and partial block reading.

\subsection{Hypergraph Product Codes}
In this paper, we study homological product codes and specifically hypergraph product codes. 
Hypergraph product codes are quantum CSS codes constructed from two input classical codes by taking the tensor product (Definition~\ref{def:tensor-product}) of their $1$-dimensional chain complexes representations.

\begin{definition}[Tensor Product of Chain Complexes]\label{def:tensor-product}
    Given two chain complexes $\cC, \cD$, their tensor product chain complex, which we denote $(C\otimes D)_\bullet$, has graded vector spaces defined as
    \begin{align}
        (C\otimes D)_k = \bigoplus_{p+q = k}C_p\otimes D_q.
    \end{align}
    For an element $c\otimes d\in C_p\otimes D_q$, the boundary map acts as
    \begin{align}
        \partial_{C\otimes D, k}(c\otimes d) = \partial_{C,p}(c)\otimes d + (-1)^p c\otimes \partial_{D,q}(d).
    \end{align}
    Note that $\partial_{C,p}(c)\otimes d$ lives in the space $C_{p-1}\otimes D_q$, and $c\otimes \partial_{D,q}(d)$ lives in the space $C_p\otimes D_{q-1}$, both of which are subspaces of $(C\otimes D)_{k-1}$. 
    When the field is $\mathbb{F}_2$, the sign $(-1)^p$ can be abridged. 
\end{definition}

When $p=0,1$ and $q=0,1$ for two input classical codes $\cC, \cD$, the tensor product is a length-2 chain complex.

Any CSS code requires $H_z$ and $H_x$ matrices that satisfy $H_xH_z^\top=0$. Hypergraph product codes assign $H_z$ and $H_x$ from the tensor product $(C\otimes D)_\bullet$ as follows.

\begin{definition}[Hypergraph Product Codes]
    A hypergraph product code HGP$(\cC,\cD)$ is the CSS code associated with a length-2 chain complex that is the tensor product of two length-1 chain complexes $\cC, \cD$ with boundary maps $\partial_C, \partial_D$. 
    We can write out the chain complex as follows.
    \begin{equation}
        \begin{tikzcd}[column sep=large]
            & 
            |[alias=Q1]| C_1 \otimes D_0 
              \arrow[rd, "\partial_{C}\otimes I_D"] 
            & \\
            |[alias=Z]| C_1 \otimes D_1 
              \arrow[rd, "\partial_{C}\otimes I_D"']
              \arrow[ru, "I_C\otimes \partial_{D}"]
            & 
            &
            |[alias=X]| C_0\otimes D_0
               \\
            &
            |[alias=Q2]| C_0 \otimes D_1
            \arrow[ru, "I_C\otimes \partial_{D}"']
            &
        \end{tikzcd}
    \end{equation}  
    Explicitly, the space of $Z$-checks, qubits, and $X$-checks are assigned to $C_1\otimes D_1$, $(C_1\otimes D_0)\  \oplus\ (C_0 \otimes D_1) $, and $C_0 \otimes D_0$ respectively. 
    The parity check matrices are consequently given by
    \begin{align}
        H_Z^\top  
        &= \begin{pmatrix}
            I_C\otimes \partial_{D}\\
            \partial_C \otimes I_D
        \end{pmatrix}, \\
        H_X &= (\partial_C\otimes I_D \quad I_C\otimes \partial_D).
        \label{eq:hgp-paritycheck}
    \end{align}
    These matrices satisfy the commutation relation $H_XH_Z^\top  = 0$.
    \end{definition}
    
\noindent 
The homology spaces of the tensor product complex are naturally related to the homology spaces of the component complexes.
\begin{lemma}[K\"{u}nneth Formula]\label{lem:Kunneth}
    \begin{align}
        H_k((C\otimes D)_\bullet) \cong \bigoplus_{p+q = k} H_p(\cC)\otimes H_q(\cD).
    \end{align}
\end{lemma}

\noindent 
We can think of $\cC$ and $\cD$ as describing two classical codes with parity check matrices $H_C=\partial_{C}$ and $H_D =\partial_{D}$. 
The vector space of bits for these codes are $C_1$ and $D_1$ with dimensions $n_C$ and $n_D$ respectively. 
The vector space of checks are $C_0$ and $D_0$, with dimensions $m_C$ and $m_D$ respectively. 
Suppose the code-space defined by $\ker(H_C)$ has parameters $[n_C,k_C,d_C]$ and the code-space defined by $\ker(H_D)$ has parameters $[n_D, k_D, d_D]$. 
Further, suppose the transpose codes $\ker(H_C^\top)$ and $\ker(H_D^\top)$ have parameters $[n_{C^\top} ,k_{C^\top} ,d_{C^\top}]$ and $[n_{D^\top} ,k_{D^\top} ,d_{D^\top}]$. 
Then the parameters of the hypergraph product code $\mathrm{HGP}(\cC, \cD)$ can be computed in terms of these classical code parameters, as follows.
\begin{align}
    n &= n_Cm_D + m_Cn_D, \\
    k &= k_Ck_{D^\top} + k_{C^\top} k_D, \\
    d &= \min(d_C,d_D,d_{C^\top} ,d_{D^\top} ).
\end{align}

For the purposes of having a convenient basis to discuss the logical action of our surgery gadgets on the hypergraph product code, we describe a canonical basis of logical operators~\cite{quintavalle2023partitioning,xu2024fast}.
For each linear map that corresponds to a classical code, for example, $\partial_C$, we perform row reduction such that these parity check matrices have the form 
\begin{align}
    \partial_C= (b_1^{C},b_2^C,...b_{k_C}^C, I_{m_C})
\end{align}
where $b_1^{C},...b_k^C$ are column vectors in $\mathbb{F}_2^{m_C}$. 
We can safely rewrite these maps because these elementary operations on the parity check matrices, which include permuting rows and columns and adding rows to other rows, preserve the code space.
Correspondingly, these elementary operations rearrange the generator matrix as
\begin{align}
G_C=\left(I_{k_C}\ \middle|\ \begin{matrix}
\beta_1^{C}\\
\beta_2^{C}\\
\vdots\\
\beta_{k_C}^{C}
\end{matrix}\right)
\end{align}
where $\beta_i^{C}:=(b_{i}^C)^\top$.
For brevity, we have only shown the row reduction for $\partial_C$, but we repeat this reduction for $\partial_D$, $\partial_C^\top$ and $\partial_D^\top$.

To write down the basis, let $\ell^C_i\in \mathbb{F}_2^{n_C}$ be the $i$th row of $G_C$, where $i\in[k_C]$.
Let $e_j^C$ be the $j$th unit vector of $\mathbb{F}_2^{n_C}$. 
We analogously define $\ell^D_i$, $e_j^D\in \mathbb{F}_2^{n_D}$, $\ell^{C^\top}_i$, $e_j^{C^\top}\in \mathbb{F}_2^{m_C}$, and $\ell^{D^\top}_i, e_j^{D^\top}\in \mathbb{F}_2^{m_D}$.
Then, a canonical basis of $Z$ logical operators is
\begin{align}
    \bar{Z}^{L}_{ij} &= \begin{pmatrix}
        \ell_i^C\otimes e_j^{D^\top}\\
        0
    \end{pmatrix},\quad \forall\  i\in [k_C],\ j\in [k_{D^\top}],\label{eq:canonical-left-z-logicals}\\
    \bar{Z}^{R}_{pq} &= \begin{pmatrix}
        0\\
        e_p^{C^\top}\otimes \ell_q^D
    \end{pmatrix},\quad \forall\ p\in[k_{C^\top}],\ q\in [k_D].\label{eq:canonical-right-z-logicals}
\end{align}
Here $\bar{Z}^{L}_{ij},\bar{Z}^{R}_{pq}\in \mathbb{F}_2^{n_Cm_D+m_Cn_D}$ are vectors whose non-zero entries correspond with $Z$ support on physical qubits in $(C_1\otimes D_0)$ and $(C_0\otimes D_1)$, respectively.
A canonical basis of $X$ logical operators is
\begin{align}
    \bar{X}^{L}_{ij} &= \begin{pmatrix}
        e_i^C\otimes \ell_j^{D^\top}\\
        0
    \end{pmatrix},\quad \forall i\in [k_C],\ j\in [k_{D^\top}]\\
    \bar{X}^{R}_{pq} &= \begin{pmatrix}
        0\\
        \ell_{p}^{C^\top}\otimes e_q^D
    \end{pmatrix},\quad \forall\ p\in[k_{C^\top}],\ q\in [k_D]
\end{align}
One can see that these representatives commute with all checks. 
As an example, consider operator $\bar{Z}^{\mathrm{L}}_{ij}$. 
Since this operator trivially commutes with all $Z$ checks, we only consider the $X$ checks in $C_0\otimes D_0$ as in Eq.~\eqref{eq:hgp-paritycheck}.
Since $\ell_{i}^C$ is a row in the generator matrix of $\partial_C$, $\partial_C \ell_i^C=0$ for all $i\in [k_C]$. 
Thus $(\partial_C\otimes I_D)(\ell_i^C\otimes e_j^{D^\top}) = 0$ and we see that $H_X\bar{Z}^{\mathrm{L}}_{ij} = 0.$
A similar argument follows for all other canonical logical operators.
Further, we can define pairs of $X$ and $Z$ canonical logicals that anticommute with each other while commuting with every other canonical logical.
Each of these pairs corresponds to a coordinate $(i,j)$ or $(p,q)$, namely
\begin{align}
    \{\bar{X}^{\mathrm{L}}_{ij},\bar{Z}^\mathrm{L}_{ij}\} = \{\bar{X}^{\mathrm{R}}_{pq},\bar{Z}^\mathrm{R}_{pq}\} = 0, \forall i,j,p,q.
\end{align}
If we arrange physical qubits in $C_1\otimes D_0$ and $C_0\otimes D_1$ in $n_C\times m_D$ and $m_C\times n_D$ grids, then these canonical logical representatives take on a nice geometric interpretation.
$(\bar{X}^{\mathrm{L}}_{ij},\bar{Z}^\mathrm{L}_{ij})$ are logical operators with support only on the $i$th row and $j$th column, respectively, of the $n_C\times m_D$ grid.
Likewise, $(\bar{X}^{\mathrm{R}}_{pq},\bar{Z}^\mathrm{R}_{pq})$ only have support on the $q$th column and $p$th row, respectively, of the $m_C\times n_D$ grid.
Note that by construction, $l_i^C$ and $l_j^{D^\top}$ have only one non-zero entry in their first $k_C$ and $k_{D^\top}$ entries, respectively.
Thus, when we restrict the $n_C\times m_D$ grid to the first $k_C$ rows and $k_{D^\top}$ columns, each physical qubit in this $k_C\times k_{D^\top}$ grid is in the support of a single pair of canonical logical operators.
By similar argument, the same is true for the restriction to a $k_{C^\top}\times k_D$ grid of the $m_C\times n_D$ grid.
From here on out, we refer to a \textit{row} of logical qubits as all of the logical qubits whose canonical logical operators have support on a particular row of qubits in the $k_C\times k_{D^\top}$ or $k_{C^\top}\times k_D$ grid.
The same applies for a \textit{column} of logical qubits.

To bound the distances of more general product complexes, we invoke the following result from Ref.~\cite{zeng2020minimal}.
\begin{lemma}[Equation~(51), Theorem~17 of Ref.~\cite{zeng2020minimal}]\label{lem:distance_of_product}
For two bounded chain complexes $\cC$ and $\cD$, where $\cD$ has exactly two non-trivial terms, the homological distances of the tensor product complex $(C\otimes D)_\bullet$ can be computed from the homological distances of the component complexes.
    \begin{align}\label{eq:product-dist}
        d_k(C\otimes D) = \min\left( d_{k-1}(C)d_1(D), d_k(C)d_0(D) \right).
    \end{align}
\end{lemma}
By taking the cochain complex $\coC, \coD$ and their tensor product $(C\otimes D)^\bullet$, we can apply the above lemma and obtain as a corollary
\begin{align}\label{eq:product-codist}
    d^k(C\otimes D) = \min\left( d^{k-1}(C)d^1(D), d^k(C)d^0(D) \right).
\end{align}

\section{Constant spacetime surgery on hypergraph product codes}
\label{sec:hgp-code}

In this section, we construct constant time overhead surgery operations that measure a \textit{row} or \textit{column} of $Z$ logical operators of the hypergraph product code in parallel. The case for $X$ logical operators is similar.
On a high level, for a HGP code $(C\otimes D)_\bullet$, we first apply the surgery method from Lemma~\ref{lem:graph-surgery} to the complex $\cC$ and obtain a sequence of ancillary complexes $\cGi{i}$. 
We then take the tensor products $(G[i]\otimes D)_\bullet$ and show that these are valid ancillary surgery complexes for $(C\otimes D)_\bullet$. 
If $\cGi{i}$ measures one homology representative $c_i$ from $\cC$, then $(G[i]\otimes D)_\bullet$ simultaneously measures a collection of homology representatives from $(C\otimes D)_\bullet$ of the form $c_i\otimes h$ for $h\in H_0(\cD)$. 
This collection can be thought of as the row of $Z$ operators corresponding to $c_i\in H_1(\cC)$.

More concretely, if $c_i$ is one of the vectors $l_i^{C}$ in Eq.~\eqref{eq:canonical-left-z-logicals},
then our surgery operation measures all logical operators $\bar{Z}^{\mathrm{L}}_{ij},\ \forall j\in [k_{D^{\top}}]$ in parallel.
In general, $c_i$ is a linear combination of vectors $l_i^{C}$, and the measured operators 
$\begin{pmatrix}
c_i\otimes e_j^{D^\top}\\
0
\end{pmatrix}$
are products of $\bar{Z}^{\mathrm{L}}_{ij}$ for all $j\in [m_D]$. 
In other words, the gadget $(G[i]\otimes D)_\bullet$ measures a row of $Z$ logical operators in parallel, each of which is a product of canonical $Z$ basis operators.
Similarly, we can first construct surgery systems for $\cD$ and take tensor products with $\cC$, which will then measure a column of $Z$ operators in parallel. 

We note that the ancilla complexes in lemma~\ref{lem:graph-surgery} can be viewed as weight-reducing~\cite{hastings2017weight} a logical operator which has been added to the stabilizer group.
From this perspective, our ancilla complexes for fast surgery can be viewed as weight-reduction of a \textit{logical codeword} of a classical code in $\cC$, followed by taking the hypergraph product with an un-tampered classical code $\cD$.
We would like to point out that~\cite{sabo2024weight-reduced} also uses weight-reduction on classical codes and to reduce the weight of \textit{checks} in two classical codes $\tilde\cC$ and in $\tilde\cD$, before taking the hypergraph product of these codes, $(\tilde C\otimes \tilde D)_\bullet $.
However, the motivation for the construction in~\cite{sabo2024weight-reduced} is to arrive at quantum codes with lower stabilizer weight and desirable $[[n,k,d]]$ parameters for small-to-medium size codes.

\subsection{Construction}
\label{sec:construction}

Consider a two-dimensional hypergraph product code $(C\otimes D)_\bullet$ with distance $d$. 
Let $c_1, \cdots, c_t\in C_1$ be an arbitrary collection of linearly independent vectors in $\ker(\partial_C)$. 
Let $\cGi{i}$ be the three-term chain complex we obtain by applying Lemma~\ref{lem:graph-surgery} to $\cC$ and $c_i$, with chain maps $g[i]:\cGi{i}\rightarrow \cC$ as in Equation~\eqref{eq:chain-map}. 
The sequence of ancilla complexes we use for surgery on the quantum code $(C\otimes D)_\bullet$ is simply $\cAi{1} = (G[1]\otimes D)_\bullet, \cdots, \cAi{t} = (G[t]\otimes D)_\bullet$.
The chain maps are 
\begin{align}
    f[i] = g[i]\otimes \id_D: (G[i]\otimes D)_\bullet\rightarrow (C\otimes D)_\bullet.
\end{align}
These are valid chain maps because $g[i]$ are valid chain maps, and the identity map $\id_D$ enables commutation of chain diagrams.

\subsection{Logical Action}
\label{sec:logical-action}

To understand the logical action of these surgery operations, we analyze the homology spaces of the coned codes. 
We first state the following helpful lemma, which we prove in Appendix~\ref{apdx:proofs}.
\begin{restatable}{lemma}{coneproduct}\label{lem:cone_product} 
Let $\cA, \cC, \cD$ be chain complexes, and let $f:\cA\rightarrow \cC$ be a chain map, which induces a second chain map $f\otimes \id_{D}: (A\otimes D)_\bullet\rightarrow (C\otimes D)_\bullet$. 
Then 
\begin{align}\label{eq:cone_product}
    \cone(f\otimes \id_{D}) \cong (\cone(f)\otimes D)_\bullet.
\end{align}
\end{restatable}
\noindent Fix $i\in [t]$, we now consider the deformed code 
\begin{align}
    \cone(g[i]\otimes \id_D) = (\cone(g[i])\otimes D)_\bullet,
\end{align}
whose homology spaces can be understood using the K\"{u}nneth formula (Lemma~\ref{lem:Kunneth}).
Specifically, we have 
\begin{equation}\label{eq:coned-homology}
    \begin{split}
    H_1((\cone(g[i])\otimes D)_\bullet) &= H_1(\cone(g[i]))\otimes H_0(\cD) \\
    &\oplus H_0(\cone(g[i]))\otimes H_1(\cD).
    \end{split}
\end{equation}
We compare this with
\begin{align}
    H_1((C\otimes D)_\bullet) &= H_1(\cC)\otimes H_0(\cD) \oplus H_0(\cC)\otimes H_1(\cD).
\end{align}

\begin{proposition}
    $H_1(\cone(g[i]))\cong H_1(\cC)/\{c_i\}$, and $H_0(\cone(g[i]))\cong H_0(\cC)$.
    \label{prop:coned-homology}
\end{proposition}
\begin{proof}
The complex $\cone(g[i])$ has the following form.
\begin{equation}
\begin{tikzcd}[column sep=large,row sep=large]
|[alias=A1]| G[i]_1 
\arrow[r, "\partial_{G\lbrack i\rbrack,1}"] 
\arrow[rd, "g\lbrack i\rbrack_1"'] &
|[alias=A0]| G[i]_0 
\arrow[r, "\partial_{G\lbrack i\rbrack,0}"] 
\arrow[rd, "g\lbrack i\rbrack_0"'] &
|[alias=Am1]| G[i]_{-1} \\
|[alias=C2] |0 \arrow[r, "0"'] &
|[alias=C1]| C_1 \arrow[r, "\partial_{C}"'] &
|[alias=C0]| C_0
\end{tikzcd}
\label{eq:cone-Gi}
\end{equation}
For the context of this proof, we let $\partial_{\cone, 1}, \partial_{\cone, 0}$ denote the boundary maps of $\cone(g[i])$. 
To show the first claim, we utilize the fact that $H_0(\cGi{i}) = 0$, as stated in condition~\ref{surgery-cond:no-gauge-qubits} of Lemma~\ref{lem:graph-surgery}. 

Take a chain $(e, b)\in \ker(\partial_{\cone, 0})$, where $e\in G[i]_0$ and $b\in C_1$. The following equation must hold.
\begin{align}
    \begin{pmatrix}
        \partial_{G\lbrack i\rbrack,1} & 0\\
        g\lbrack i\rbrack_0 & \partial_C
    \end{pmatrix}
    \begin{pmatrix}
        e\\
        b
    \end{pmatrix}
    =
    \begin{pmatrix}
        \partial_{G\lbrack i\rbrack,1} (e)\\
        g\lbrack i\rbrack_0  (e) + \partial_C (b)
    \end{pmatrix}
    =\vec{0}
\end{align}
Consequently, $e$ must be in $\ker(\partial_{G\lbrack i\rbrack,0})$, which is also $ \im(\partial_{G\lbrack i\rbrack,1})$ because the 0-homology is trivial.
This means there is an element $v\in G[i]_1$ such that 
\begin{align}
    \partial_{\cone, 1}((v, 0)) + (e, b) = (0, b') \in \ker(\partial_{\cone, 0}),
\end{align}
for some $b'\in C_1$. 
$(0, b')$ must be in $\ker(\partial_{\cone, 0})$ because $(e, b)\in \ker(\partial_{\cone, 0})$ by assumption, and $\partial_{\cone,0}\partial_{\cone, 1}((v, 0))=0$ because $\cone(g[i])$ is a valid chain complex.
Note that $(e, b)$ and $(e, b')$ are in the same homology class in $H_1(\cone(g[i]))$.
$(0,b')\in\ker(\partial_{\cone,0})$ implies $b'$ must be in $\ker(\partial_C)$, and thus every class in $H_1(\cone(g[i]))$ must correspond to an independent class in $H_1(\cC)$. 

It remains for us to consider how the classes in $H_1(\cC)$ change due to the added space $G[i]_1$. As stated in Lemma~\ref{lem:graph-surgery}, $\dim(\ker(\partial_{G[i], 1})) = 1$, and the element $v\in \ker(\partial_{G[i], 1})$ satisfy $g[i]_1(v) = c_i$. 
Therefore, the class $\{c_i\}\in H_1(\cC)$ is now in $\im(\partial_{\cone, 1})$ and removed from $H_1(\cone(g[i]))$, while the rest of the classes remain independent. 
This proves our first claim. 

To show the second claim, $H_0(\cone(g[i])) \cong H_0(\cC)$, we use the fact that $g[i]$ is a chain map and $H_{-1}(\cGi{i}) = 0$.
Take $(r, h)\notin \im(\partial_{\cone, 0})$, where $r\in G[i]_{-1}$ and $h\in C_0$. Since $H_{-1}(\cGi{i}) = 0$, there is $e\in G[i]_0$ such that
\begin{align}
    \partial_{\cone, 0}((e,0)) + (r, h) = (0, h') \in H_0(\cone(g[i])).
\end{align}
For $(0, h')$ to be in $H_0(\cone(g[i]))$, $h'$ must be in $H_0(\cC)$. 
Therefore every class in $H_0(\cone(g[i]))$ must correspond to a class in $H_0(\cC)$. 
Under this mapping, chains in independent classes in $H_0(\cone(g[i]))$ are mapped to chains in independent classes in $H_0(\cC)$.

Conversely, take $h\in H_0(\cC)$, then $(0, h)$ must be in $H_0(\cone(g[i]))$. Assume for the sake of contradiction that $(0, h)\in \im(\partial_{\cone, 0})$, then there must be $(e, b)$ for $ e\in G[i]_0, b\in C_1,$, such that
\begin{align}
    \partial_{\cone, 0}((e, b)) = (0, h).
\end{align}
In particular, this means $\partial_{G[i], 0}(e) = 0$, which means $e = \partial_{G[i], 1}(v)$ for some $v\in G[i]_1$. Since $g[i]$ is a chain map, the diagram in equation~\eqref{eq:cone-Gi} commutes and we have
\begin{align}
    g[i]_0 \circ \partial_{G[i], 1} = \partial_C \circ g[i]_1.
\end{align}
Take $b' = g[i]_1(v)$, we see that 
\begin{align}
    \partial_{\cone, 0}((0, b+b')) = (0, \partial_C(b+b')) = (0, h),
\end{align}
which means $h\in \im(\partial_C)$, a contradiction.
Therefore every class in $H_0(\cC)$ gives an independent class in $H_0(\cone(g[i]))$.
We conclude that $H_0(\cone(g[i])) \cong H_0(\cC)$.
\end{proof}

\begin{remark}
    The cone $\cone(g[i])$ is motivated by applying gauging or homological measurements (Lemma~\ref{lem:graph-surgery}) to a classical code $\cC=0\xrightarrow{0}C_1\xrightarrow{\partial_C} C_0$. The ancilla complex that facilitates the gauging measurement of a codeword $c_i\in H_1(\cC)$ is the three-term complex $\cGi{i}$ such that the product of all elements in $G[i]_{1}$ is $c_i$.
    Geometrically, the three-term complex $\cGi{i}$ in $\cone(g[i])$ can be interpreted as a complex between vertices $G[i]_1$, edges $G[i]_{0}$, and faces $G[i]_{-1}$ attached to the bits and checks of the classical code $\cC$~\cite{williamson2024low-overhead}. 
    
    Naturally, as $\cGi{i}$ has three non-trivial terms, the tensor product of $\cGi{i}$ with $\cD=D_1\xrightarrow{\partial_{D}}D_0$, has four terms instead of three in the base code, $(C\otimes D)_\bullet$. Crucially, we designate one of these terms, $A_2 = G[i]_1\otimes D_1$, to be the ancilla complex meta-checks in the mapping cone for fast surgery (Eq.~\eqref{eq:fast-surgery}). 
\end{remark}

Proposition~\ref{prop:coned-homology} and equation~\eqref{eq:coned-homology} implies that in the deformed code $\cone(g[i]\otimes \id_D)$, the logical qubits of $(C\otimes D)_\bullet$ which correspond to the homology classes $\{c_i\}\otimes H_0(\cD)$ are now in $\im(\partial_{\cone, 1})$; 
equivalently, these $Z$-logical operators are now a product of newly introduced $Z$-stabilizers in the ancilla system $\cAi{i}$. 
in the context of code surgery, these logical operators are now measured simultaneously.
The sequences of surgery operations given by $\cAi{1}, \cdots, \cAi{t}$ then measures the operators corresponding to $\{c_1\}\otimes H_0(\cD), \cdots, \{c_t\}\otimes H_0(\cD)$ in order. 

While simultaneous measurements of multiple operators can be highly efficient, ideally we would also like to construct surgery operations that are both fast and addressable, i.e., capable of targeting individual logical operators. 
A natural idea here is to consider changing the ancilla complexes from $(G[i]\otimes D)_\bullet$ to $(G[i]\otimes D')_\bullet$ for some other complex $\cD'$ which is homomorphic to $\cD$. 
For instance, $\cD'$ may be obtained from puncturing or augmenting $\cD$, similar to the constructions in Ref.~\cite{xu2024fast}.
A technical challenge with this approach, however, is that our central Lemma~\ref{lem:cone_product} no longer applies and new proofs would be needed to analyze the coned complexes and the compacted code. 
We leave this promising direction for future work.

We further note that our constructions are closely related to the homomorphic measurement gadgets~\cite{huang2023homomorphic} constructed for homological product codes in Ref.~\cite{xu2024fast}. 
We briefly recall their constructions.
Specifically, Ref.~\cite{xu2024fast} also constructed $\cC', \cD'$ with homomorphic chain maps $f, g$ to $\cC, \cD$ respectively. 
They treated $(C'\otimes D')_\bullet$ as an ancillary quantum code, initialized it into a logical product state, and applied a collection of physical CNOTs gates (with control qubits in $(C'\otimes D')_\bullet$ and target qubits in $(C\otimes D)_\bullet$) specified by the product chain map $f\otimes g$.
This physical CNOT circuit (henceforth called the ``homomorphic CNOT'') induces a logical CNOT circuit between the logical qubits of $(C'\otimes D')_\bullet$ and that of $(C\otimes D)_\bullet$. 
By measuring the physical qubits in $(C'\otimes D')_\bullet$ transversally, they enact a logical measurement of certain operators (dependent on $\cC'$ and $\cD'$) in the base code $(C\otimes D)_\bullet$. 

The critical difference between our methods and theirs is reflected in the time overheads of our respective protocols. In this work, our surgery gadgets incur constant time overhead in amortization. 
In Ref.~\cite{xu2024fast}, while the homomorphic CNOT takes constant time, each round of logical measurement requires a well-prepared logical state of the code $(C'\otimes D')_\bullet$. 
For generic HGP codes, this state preparation requires $O(d)$ rounds of syndrome measurements. 
Readers familiar with the work of Ref.~\cite{zhou2024algorithmic} may think that the techniques of algorithmic fault-tolerance may be applied, where we prepares each state for only $O(1)$ rounds, resulting in ill-prepared states to be used in the measurements. 
The hope is that by performing many such measurements in succession and decoding all of them together (similar to our fast surgery formulation), the overall protocol would be fault-tolerant.
This approach, however, does not work in general~\cite{harry-personal-communication}. 
Therefore, on face value, the gadgets constructed in this work are indeed much faster.

Looking deeper, part of the reason algorithmic fault-tolerance cannot be directly applied to gadgets in Ref.~\cite{xu2024fast} is, in fact, that the compacted code (obtained by treating the ancillary complexes $(C'\otimes D')_\bullet$ as surgery gadgets instead of ancillary code blocks, see footnote~\footnote{More specifically, when we treat the complex $(C'\otimes D')_\bullet$ as ancillary code blocks, we are applying a CNOT circuit according to the chain map $f\otimes g$. When we treat the complex as a surgery gadget, we are taking the mapping cone $\cone(f\otimes g)$ as the deformed code, which in particular means that the complex $(C'\otimes D')_\bullet$ has its degrees shifted by one. The compacted code in this context is obtained from Definition~\ref{def:compacted-code} by treating a sequence of ancillary complexes $(C'\otimes D')_\bullet$, which are originally used for homomorphic measurement in Ref.~\cite{xu2024fast}, as fast surgery gadgets.})
has distance lower than $d$ in general~\cite{harry-personal-communication}.
Therefore, the true difference between our work and Ref.~\cite{xu2024fast} is that we constructed ancillary complexes and chain maps which ensures the compacted code has high distance.
This additional condition is what enabled constant time overhead surgery.

To conclude this discussion, we note that Ref.~\cite{cowtan2025fast} showed a circuit-equivalence between CSS surgery and homomorphic measurements via ZX-calculus. 
The existence of such an equivalence is not surprising; after all, they are both based on homomorphic chain maps.
Following this equivalence, we believe our constructions can also be applied as homomorphic measurement gadgets, where the ancilla code blocks are prepared for $O(1)$ rounds each, to achieve the same measurements fault-tolerantly.

\subsection{Compacted Code Distance}
\label{sec:compacted-code-distance}
We first recall the definition of the compacted code (Def.~\ref{def:compacted-code}). 
To construct the compacted code, we consider the trivial extensions of the chain maps $f[i]:\cAi{i}\rightarrow (C\otimes D)_\bullet$ to 
\begin{align}
    f'[i]: \bigoplus_{i\in [t]}\cAi{i}\rightarrow (C\otimes D)_\bullet,
\end{align}
whose action is the same on $\cAi{i}$ and trivial on the rest of the spaces $\cAi{j}$, $j\ne i$. 
The domain of $f'[i]$ is the direct sum of all ancilla complexes in $t$ sequential fast measurements, mapped onto the base code $(C\otimes D)_\bullet$. 
Note that this is equivalent to taking a similar extension for the maps $g[i]$, namely
\begin{align}
    g'[i]:\bigoplus_{i\in [t]}\cGi{i}\rightarrow \cC, 
\end{align} 
and setting $f'[i] = g'[i]\otimes \id_D:\bigoplus_{i\in [t]} (G[i]\otimes D)_\bullet \rightarrow (C\otimes D)_\bullet$. 
Further define 
\begin{align}
    \bar{g}:= \sum_{i\in[t]} g'[i], \quad \bar{f}:= \sum_{i\in[t]} f'[i] = \bar{g}\otimes \id_D.
    \label{eq:gbar-fbar}
\end{align}
$\bar{g}$ and $\bar{f}$ is simply the sum of all extended chain maps over $t$ surgeries, such that $\bar{f}$ acts non-trivially as $f[i]$ on $\cAi{i},\ \forall i\in [t]$.
The compacted code is then defined as the mapping cone of $\bar{f}$.
\begin{align}
    \cc_\bullet := \cone(\bar{f}) = \cone(\bar{g}\otimes \id_D) = (\cone(\bar{g})\otimes D)_\bullet.
\end{align}
The last equality above follows from Lemma~\ref{lem:cone_product}.
The complex $\cone(\bar{g})$ has the following form,
\begin{equation}
\begin{tikzcd}[column sep=large,row sep=large]
|[alias=A1]| \oplus_{i\in [t]} G[i]_1 
\arrow[r, "\partial_{G,1}"] 
\arrow[rd, "\bar{g}_1"'] &
|[alias=A0]| \oplus_{i\in [t]} G[i]_0 
\arrow[r, "\partial_{G,0}"] 
\arrow[rd, "\bar{g}_0"'] &
|[alias=Am1]| \oplus_{i\in [t]} G[i]_{-1} \\
|[alias=C2] |0 \arrow[r, "0"'] &
|[alias=C1]| C_1 \arrow[r, "\partial_{C}"'] &
|[alias=C0]| C_0
\end{tikzcd}
\label{eq:cone-G}
\end{equation}
In the above expression, we have
\begin{align}
    \bar{g}_1 = \sum_{i\in [t]} g'\lbrack i\rbrack_1, &\quad 
    \bar{g}_0 = \sum_{i\in [t]} g'\lbrack i\rbrack_0, \\
    \partial_{G,1} = \bigoplus_{i\in [t]}\partial_{G\lbrack i\rbrack,1},
    &\quad
    \partial_{G,0} = \bigoplus_{i\in [t]}\partial_{G\lbrack i\rbrack,0}.
\end{align}

\begin{proposition}\label{prop:1D-cone-compacted-homology}
    $H_0(\cone(\bar{g}))\cong H_0(\cC)$, and 
    \begin{align}
        H_1(\cone(\bar{g}))\cong H_1(\cC)/\Span\{\{c_1\},\cdots, \{c_t\}\}.
    \end{align}
\end{proposition}
\begin{proof}
    Follows from iteratively applying the arguments in Proposition~\ref{prop:coned-homology}.
\end{proof}
\begin{proposition}\label{prop:compacted-cone-0-codistance}
    $d^0(\cone(\bar{g})) \ge d$.
\end{proposition}
\begin{proof}
    Consider $(h, r)\in H^0(\cone(\bar{g})) = \ker\left(\delta_{\cone(\bar{g})}^0 \right)$, where $h\in C_0$ and $r\in \oplus_{i\in [t]} G[i]_{-1}$. Since $H^0(\cGi{i}) = 0$ for all $i$, we must have $h\ne 0$, which means $h\in \ker(\delta_C)$. This implies $|h|\ge d$.
\end{proof}

\begin{restatable}{proposition}{CompactedConeDistance}\label{prop:1D-cone-compacted-distance}
    $d^1(\cone(\bar{g})) \ge 1$, and $d_1(\cone(\bar{g})) \ge d$.
\end{restatable}

We prove this proposition in Appendix~\ref{apdx:proofs}, and note that the proof is a fairly standard expansi on argument, similar to those developed in Refs.~\cite{cross2024improved,williamson2024low-overhead,Ide_2025,swaroop2024universal}. 
Intuitively, the argument considers a logical operator in $H_1(\cone(\bar{g}))$ which we do not measure. Each of these logical operators corresponds to a logical supported entirely in $C_1$, the bits of the original code, which must have weight at least $ d$ by assumption. To see that the distance is preserved, we deform a logical $c_j\in C_1$ with an arbitrary set of checks $v\in \oplus_{i\in[t]} G[i]_1$ and prove that the weight across $C_1\oplus (\bigoplus_{i\in[t]}G[i]_0)$ is at least $ d$. Given the expansion condition~\ref{surgery-cond:expansion} in Lemma~\ref{lem:graph-surgery} of the map $\partial_{G[i],1}:G[i]_1 \rightarrow G[i]_0$, one can apply this to the extended field and boundary of the compacted code, $\partial_{G,1}:\bigoplus_{i\in [t]}G[i]_1 \rightarrow \bigoplus_{i\in [t]}G[i]_0$, and to show that adding support from $\partial_{G,1}(v)$ does not reduce the weight of the logical.

\begin{proposition}
    The compacted code defined by the ancillary complexes $\cAi{i}$ and chain maps $f[i]$ has 1-systolic and 1-cosystolic distances at least $d$.
    \label{prop:compacted-code-dist}
\end{proposition}
\begin{proof}
From Lemma~\ref{lem:distance_of_product}, equation~\eqref{eq:product-codist}, Propositions~\ref{prop:compacted-cone-0-codistance} and~\ref{prop:1D-cone-compacted-distance}, we can bound the distances of the compacted code.
\begin{align}
    &d^1(\cc_\bullet) \notag\\
    &= \min(d^1(\cone(\bar{g}))d^0(\cD), d^0(\cone(\bar{g}))d^1(\cD)) \\
    &\ge \min(1\times d, d\times 1) = d. \\
    &d_1(\cc_\bullet) \notag\\
    &= \min(d_1(\cone(\bar{g}))d_0(\cD), d_0(\cone(\bar{g}))d_1(\cD)) \\
    &\ge \min(d\times 1, 1\times d) = d. 
\end{align}
We see that condition~\ref{fast-cond:compacted-distance} of Lemma~\ref{lem:fast-surgery-conditions} is satisfied.
\end{proof}

\subsection{Meta-check Distance}
\label{sec:meta-check-dist}
Finally, we show that the ancillary complexes has high meta-check distance.
\begin{proposition}
    Each ancillary complex $\cAi{i} = (G[i]\otimes D)_\bullet$ satisfy conditions~\ref{fast-cond:metacheck} and~\ref{fast-cond:non-overlap} of Lemma~\ref{lem:graph-surgery}.
\end{proposition}
\begin{proof}
    To show that $d^1(\cAi{i})\ge d$, we use Lemma~\ref{lem:distance_of_product} and equation~\eqref{eq:product-codist}.
    \begin{align}
        &d^1((G[i]\otimes D)_\bullet) \notag \\
        &= \min(d^1(\cGi{i})d^0(\cD), d^0(\cGi{i})d^1(\cD)) \\
        &= \min(1 \times d, \infty\times 1).
        \label{eq:metacheck-distance}
    \end{align}
    Here we follow the convention in Ref.~\cite{zeng2020minimal} and set $d^0(\cGi{i})$ to in $fty$ since the cohomology space is empty, as stated in Lemma~\ref{lem:graph-surgery}.
    Condition~\ref{fast-cond:non-overlap} follows from the fact that $f[i] = g[i]\otimes \id_D$ and $g[i]_1$ is $1$-sparse.
\end{proof}

Putting everything together, we obtain the main result of this paper.
\begin{theorem}\label{thm:main}
    Given a $[[n,k,d]]$ hypergraph product code $(C\otimes D)_\bullet$, we can construct a sequence of ancillary complexes $\cAi{1}, \cdots, \cAi{t}$ which satisfy the fast surgery conditions in Lemma~\ref{lem:fast-surgery-conditions}.
    Performing fast surgery with these complexes, as formulated in Section~\ref{sec:fast-surgery}, lets us measure rows (and similarly columns) of $Z$ logical operators (and similarly $X$ operators) in parallel, with each surgery operation taking $O(1)$ times in amortization.
    Furthermore, each complex utilizes $\tilde{O}(n)$ physical ancilla qubits. 
    Therefore the overall spacetime overhead of these logical gates is $\tilde{O}(1)$.
\end{theorem}

The space overhead of each ancillary complex is
\begin{align}
    &\dim(A[i]_0) \notag \\
    &=\dim(G[i]_0)\dim(D_0)+\dim(G[i]_{-1})\dim(D_1)\\
    &\le O(|c_i|\log^3(|c_i|))\cdot 2n_D.
\end{align}
For the inequality, we have used condition~\ref{surgery-cond:size-bound} of lemma~\ref{lem:graph-surgery} to bound the size of $G[i]_0$ and $G[i]_{-1}$ in terms of $|c_i|$, the measured operator of $\cone(g[i])$.
Suppose we take a $[[n,O(n),\sqrt{n}]]$ hypergraph product code constructed from a classical codes $\cC$ and $\cD$ each with parameters $[O(\sqrt{n}),O(\sqrt{n}),O(\sqrt{n})]$. 
It must be that $|c_i|\le O(\sqrt{n})$, and thus it follows that  
\begin{align}
    \dim(A[i]_0)&\le O(\sqrt{n}\log^{3}(\sqrt{n}))\cdot O(\sqrt{n}) \\
    &=O(n\log^{3}(n)).
\end{align}
We note that the above equation is an upper bound on the space overhead of each surgery gadget. 
In practice, we expect the extra poly-logarithmic factors that arise from the use of the decongestion lemma in~\cite{Ide_2025,williamson2024low-overhead} to be mostly unnecessary.
This expectation is supported by existing practical constructions of surgery gadgets, see for instance Refs.~\cite{cowtan2024ssip,cross2024improved,williamson2024low-overhead,ide2024fault,yoder2025tour}.
Lastly, we remark that these surgery gadgets have a non-zero threshold for LDPC hypergraph product codes with distances $d\ge O(n^a)$ for $a>0$, because the surgery gadgets of Theorem~\ref{thm:main} keep all checks LDPC and preserve the distance of the base code.
The existence of threshold follows from standard counting arguments, see for instance Theorem~1 of Ref.~\cite{Kovalev_2013}, Section~5 of Ref.~\cite{gottesman2014fault-tolerant}, and Section~6.3 of Ref.~\cite{he2025composable}.

\section{Examples and Applications}
\label{sec:examples_applications}

At a high level, we construct our constant-time surgery gadget by lifting a one-dimensional surgery ancilla system to a two-dimensional ancilla system via homological product. 
This recipe is broadly applicable and can be easily generalized. 
In this section, we showcase a few interesting applications and extensions of our construction for various logical processing tasks. 

\paragraph{Choice of 1D surgery gadgets.}
In Section~\ref{sec:construction}, we took the standard result from Refs.~\cite{williamson2024low-overhead,Ide_2025} as our 1D surgery gadgets and lifted them by homological product into a surgery gadget for a HGP code block. 
It is straightforward to extend our proof to use other surgery gadgets~\cite{swaroop2024universal,cowtan2025parallel,he2025extractors,zheng2025high} as the 1D component.
As an example, we consider the universal adapters of Ref.~\cite{swaroop2024universal}, which can be lifted by homological product to perform parallel entangling measurements between multiple different hypergraph product codes, as long as they share at least one classical code.

More formally, consider $m$ blocks of hypergraph product codes which share one classical component code,
\begin{align}
    &(C^{(1)}\otimes D)_\bullet\oplus (C^{(2)}\otimes D)_\bullet\oplus...(C^{(m)}\otimes D)_\bullet \\&= \bigoplus_{j\in[m]}(C^{(j)}\otimes D)_\bullet = ((\bigoplus_{j\in[m]} C^{(j)})\otimes D)_\bullet
    \\&=(C'\otimes D)_\bullet.
    \label{eq:many-hgp-codes}
\end{align}

\noindent Note that $C^{(i)}$ and $C^{(j)}$ can be arbitrary, different codes (with the constraint that all the HGP codes have distance at least $d$).
In the last line, $(C'\otimes D)_\bullet$ is the product of $\cD$ with 
\begin{align}
    \cC' = (\bigoplus_{j\in[m]} C_1^{(j)})\xrightarrow{\partial_C'} (\bigoplus_{j\in[m]} C_0^{(j)}), \partial_C'= \bigoplus_{j\in [m]}\partial _{C}^{(j)}.
\end{align}
Similar to the one code block case, we can measure in parallel the logical operators $c_i\otimes H_0(\cD)$ where $c_i=(c_i^{(1)},c_i^{(2)},...c_i^{(m)})\in H_1(\cC')$ with $c_i^{(j)}\in H_1(\cC^{(j)})$ for all $j$.
This measurement needs a one-dimensional surgery gadget that measures $c_i$, which can be built using the universal adapters of Ref.~\cite{swaroop2024universal}. 
In more detail, one can use Theorem 11 of Ref.~\cite{swaroop2024universal} to construct a surgery system $\cGi{i}$ that measures $c_i$ given ancillary complexes that measure $c_i^{(1)},...,c_i^{(m)}$ individually (which can be built using Lemma~\ref{lem:graph-surgery}). 
Such a construction has the advantage of being more modular and flexible, compared to building one full system directly with Lemma~\ref{lem:graph-surgery}.
The rest of the construction follows with this new $\cGi{i}$.  
As a result, we can measure logical operators and thereby entangle logical qubits between two or more non-identical hypergraph product codes, given that they share a common base classical code $\cD$.

We note an important technical detail here: the constructions in Refs.~\cite{swaroop2024universal} utilize a more relaxed notion of expansion called \textit{relative expansion}, while our Proposition~\ref{prop:1D-cone-compacted-distance} utilizes the stronger expansion guarantee given by Lemma~\ref{lem:graph-surgery}. 
This notion of relative expansion is later utilized in the extractor architecture~\cite{he2025extractors} as well. 
For technical correctness, we prove a version of Proposition~\ref{prop:1D-cone-compacted-distance} using relative expansion in Appendix~\ref{apdx:proofs}.
(see Prop.~\ref{prop:dist_relative_exp}).

So far we have restricted each one-dimensional surgery gadget to be measuring one logical operator (or codeword) from the base classical code. 
By careful choice of a $G_\bullet$ that describes a hypergraph instead of a graph, one can measure several codewords of the classical code $C_\bullet$ simultaneously~\cite{williamson2024low-overhead}, say $c_I\in H_1(\cC)$ for $I\in[i,i+1,...,i+p]$.
Then, $(\cone(g)\otimes D)_\bullet$ will measure even more logical operators in parallel, $c_I\otimes h$ for all $I$ and $h\in H_0(\cD)$. 
The challenge with this approach is to construct low space overhead hypergraphs for parallel measurements. 
We leave this interesting direction to future work, and note that parallel surgery have been explored theoretically in Refs.~\cite{zhang2025time-efficient,cowtan2025parallel}, and heuristically in Ref.~\cite{zheng2025high}.

\paragraph{Higher-dimensional homological product and HGP codes.}

Our gadgets have a natural generalization to higher-dimensional homological product codes.
In this setting, we can measure `lines' or `planes' of logical operators depending on the specifics of how the higher-dimensional product code is constructed.
As a concrete example, take a 3D homological product code, $(Q\otimes C)_\bullet $, a tensor product of a quantum $Q_\bullet:Q_2\xrightarrow[]{\partial_{Q,2}}Q_1\xrightarrow[]{\partial_{Q,1}}Q_0$ and a classical code $C_\bullet$.
Suppose that we associate the space $Q_0\otimes C_0$ for $X$ checks, and we would like to measure logical operators of the form $L\otimes h$, where $h\in H_0(\cC)$. 
One can loosely interpret these $L\otimes h$ as a line of logical operators that we will measure in parallel.
The main idea of this generalization is to lift the results for the compacted coned code in Section~\ref{sec:hgp-code} to a compacted coned code that performs surgery on a quantum code.
We can apply the results of Lemma~\ref{lem:graph-surgery} to measure a logical operator $L\in H_1(Q_\bullet)$, since $Q_\bullet$ is a three-term chain complex.
This creates the coned code, $\cone(q[i])$, where $q[i]:\cA[i]\rightarrow Q$.
The deformed code that performs fast surgery is the tensor product code of $\cone(q[i])$ with the classical code $\cC$.
The same arguments of preserved distance of the compacted cone follow from Propositions~\ref{prop:compacted-cone-0-codistance} and~\ref{prop:1D-cone-compacted-distance} because their proofs only use the homology and expansion conditions of $\cone(q[i])$, which are guaranteed by Lemma~\ref{lem:graph-surgery}.
One can invoke equation~\eqref{eq:product-dist} to argue that the final compacted code which is a tensor product of the compacted cone with a classical code, will also have high distance.

We can similarly measure planes of logical operators in parallel for higher-dimensional hypergraph product codes. 
In contrast to general homological product codes, these codes are products of many classical codes.
Suppose we take the hypergraph product of three classical codes, $\cC:C_1\rightarrow C_0$,  $\cD:D_1\rightarrow D_0$, and  $E_\bullet:E_1\rightarrow E_0$. 
Let us assign X checks to $C_0\otimes D_0\otimes E_0$, qubits to $(C_1\otimes D_0\otimes E_0)\oplus (C_0\otimes D_1\otimes E_0) \oplus (C_0\otimes D_0\otimes E_1)$, $Z$ checks to $(C_1\otimes D_1\otimes E_0)\oplus (C_1\otimes D_0\otimes E_1)\oplus (C_0\otimes D_1\otimes E_1)$, and finally, meta checks on $Z$ checks to $C_1\otimes D_1\otimes E_1$.
An example of a plane of logicals that we can measure is $c_i\otimes h_D\otimes h_E$, where $h_D\in H_0(\cD)$ and $h_E\in H_0(E_\bullet)$, corresponds to measuring, in parallel, $Z$ logical operators whose canonical representatives~\cite{xu2024fast} have support on a plane with dimensions $k_C\times k_{D^\top }\times k_{E^\top }$ grid in $C_1\otimes D_0\otimes E_0$.
Using the same techniques in Section~\ref{sec:hgp-code}, we can make a 1-dimensional coned code $\cone(g)$ that measures a logical operator of the classical code $\cC$, $c_i\in H_1(\cC)$.
The deformed code is the tensor product of this coned code $\cone(g)$, with the other two classical codes, $\cD$ and $E_\bullet$. 
From the application of Lemma~\ref{lem:distance_of_product} and Lemma~\ref{lem:Kunneth}, our claims on the compacted code homology and distance (Prop.~\ref{prop:compacted-code-dist}) generalize because we are only taking another tensor product with code $E_\bullet$.

\paragraph{Batched computation.}

On 2D HGP codes, our gadgets perform parallel measurement on rows (or columns) of logical qubits. 
They currently do not have the flexibility to measure more complicated patterns of qubits.
Nevertheless,
we can use these row and column-wise parallel operations to run copies of the same logical circuit, following the batched computation perspective from Ref.~\cite{xu2025batched}.
Batched execution of logical circuits involves running copies of the same logical circuit in parallel.
The main benefit is that fault-tolerantly performing the same operation or logical measurement on many copies of logical qubits is cheaper than doing the same number of operations serially.
We can apply this perspective by using our row and column-wise measurements to perform the same $X$ or $Z$ measurement on many logical qubits in parallel, where each measurement occurs on a separate logical circuit.

For example, let the set of logical qubits used in a logical circuit we wish to perform be $\{L_i\}$, $i\in [N]$.
Suppose we want to run this circuit many times, $T$. 
The natural thing to do is to serialize these $T$ logical circuit runs.
That is, once one logical circuit has finished executing, then we run the same logical circuit again.
This means that our system must have enough qubits to encode and perform logic on $N$ logical qubits.
In this setting of serialized logical circuit runs, to make use of our fast surgery gadgets, the logical circuit we want to run needs to be compiled into uniform Pauli measurements on rows of logical operators.

Instead of running this circuit serially, one can run $k_C$ copies of the logical circuit simultaneously, on different sets of $N$ logical qubits. 
This means that we must encode and perform logic on at least $Nk_C$ logical qubits at any given time.
Let us index these logical qubits $\{L_{it}\}$ where $i\in[N]$ and $t\in [k_C]$.
Then, for fixed $i$, in other words, a fixed logical qubit in the circuit, we encode logical qubits $\{L_{it}|\ t\in [k_C]\}$ on a row of logical qubits in a $(C\otimes D)_\bullet$ hypergraph product code.
Different batches of qubits corresponding to different $i\in [N]$ can be encoded in different HGP code blocks.
By using our fast surgery gadgets to measure $Z$ operators on rows in parallel, we can execute the same $Z$ measurement on all parallel instances of the logical circuit.
This strategy opens new possibilities to use our fast surgery gadgets with fewer requirements on the structure of the logical circuit, at the cost of requiring a larger scale of physical qubits to execute the circuit instances in parallel.

In comparison to the batched computation developed in Ref.~\cite{xu2025batched}, we make two further remarks. 
One notable challenge with using our gadgets for batched computation is that the $X$ basis measurements act on columns of logical qubits, while the $Z$ basis measurements act on rows of logical qubits. 
Consequently, if we encode an instance of the logical circuit on a column of logical qubits, which enable batched $Z$-basis measurements, we cannot perform batched $X$-basis measurements directly. 
It will be interesting to explore methods to circumvent this alignment issue, such as teleporting transversal Hadamard gates into the code block. 
We note that this alignment issue does not occur in Ref.~\cite{xu2025batched} as they considered homological product of two distance-$d$ quantum codes.
In exchange, however, a challenge with implementing the batched computation in Ref.~\cite{xu2025batched} is the large number of physical qubits needed (despite the rate being constant asymptotically). 
Our gadgets can be implemented with a much smaller number of physical qubits, as they operate on as few as one 2D HGP code block.
For this reason, we believe that our fast surgery gadgets are well-suited for batched computation in practice.

\paragraph{Practical Prospects.}
Asymptotically, our gadget incurs a poly-logarithmic space overhead, which comes from the overhead in Lemma~\ref{lem:graph-surgery}. 
In practice, we expect the space overhead to be a small constant, as we have observed in several studies~\cite{cross2024improved,williamson2024low-overhead,Ide_2025,yoder2025tour,webster2025explicit}.
On the other hand, throughout this work we have not discussed the decoding task for constant-time surgery. 
We leave this question for future works to explore, and note that the decoding efficiency is critical to the practical performance of our scheme.

\section{Intuitive example: Constant spacetime overhead and addressable toric code measurements}
\label{sec:toric-code}
    
\begin{figure*}[h!]
    \centering
    \includegraphics[width=.45\linewidth]{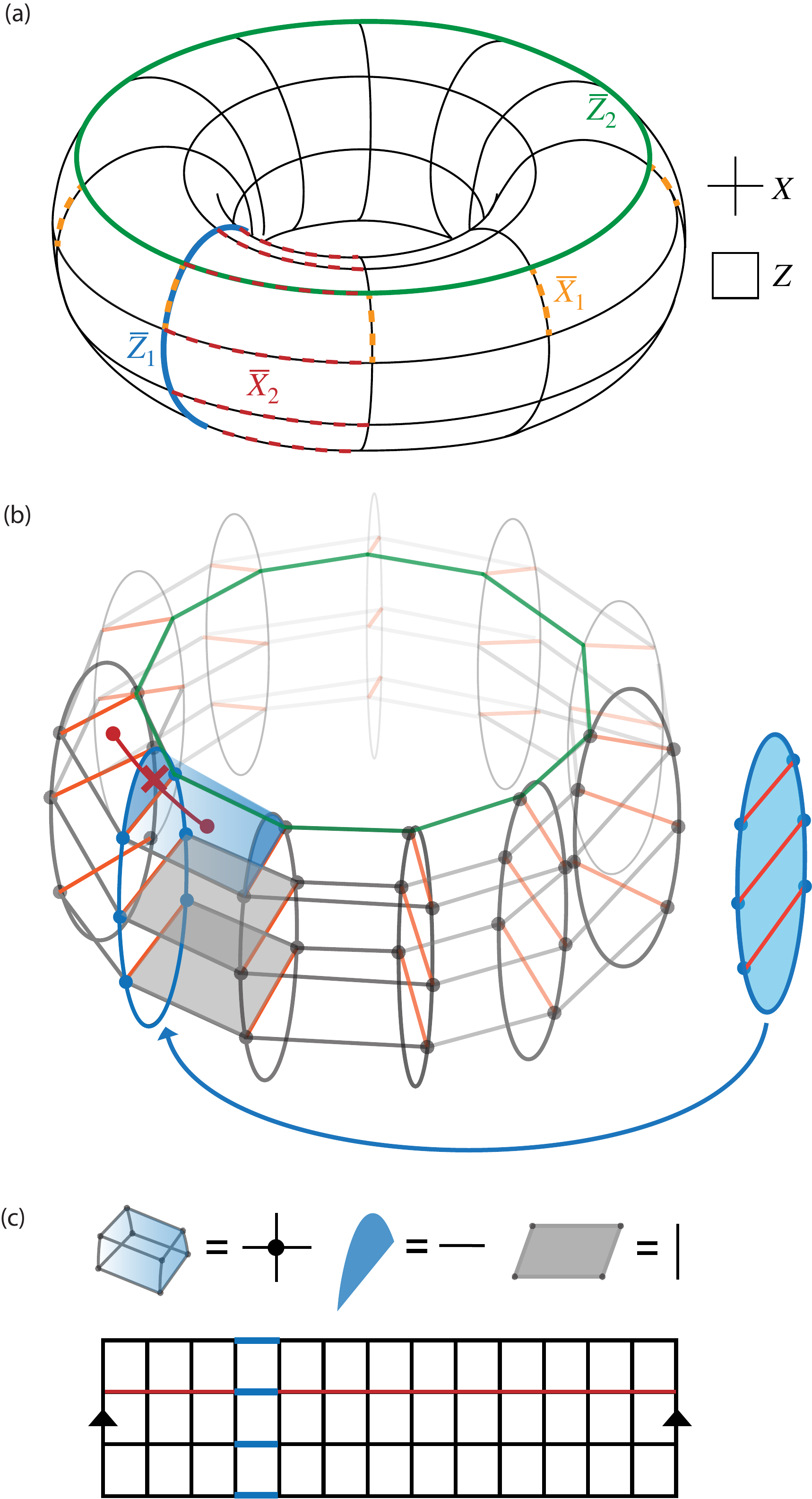}
    \caption{(a) A toric code, with qubits as edges, $X$ checks assigned to vertices, and $Z$ checks assigned to faces. 
    $\bar{Z}_1$ (solid blue), $\bar{X}_1$ (dashed orange), $\bar{Z}_2$ (solid green), and $\bar{X}_2$ (dashed red) logical operators are highlighted. 
    (b) $\cone(f[i])$ code for measuring the $\bar{Z}_1$ logical (outlined in light blue), whose minimum-weight representatives are in $ (C_1\otimes D_0)$.
    Notice that $\bar{Z}_2$, highlighted in green, remains in the homology of this filled torus and thus, unmeasured.
    Gray edges are qubits in the spaces $(C_0\otimes D_1)\oplus (C_1\otimes D_0)$, which are associated with qubits from the base toric code, $(C\otimes D)_\bullet$. 
    Orange edges are qubits in $(G_0\otimes D_0)\oplus (G_{-1}\otimes D_1)$ which are qubits from the ancilla complex $(G\otimes D)_\bullet$.
    Faces are assigned to $Z$ checks. 
    The faces on the surface of the filled torus are the $Z$ checks of the base toric code. Additional faces have been introduced, including solid blue checks that fill each minimum-weight $\bar{Z}_1$ loop. 
    These checks are in $ G_1\otimes D_0$ and multiply to equal a $\bar{Z}_1$ representative (see filled $\bar{Z}_1$ loop on the right). 
    Checks in $G_0\otimes D_1$ are highlighted in gray.
    Vertices are assigned to $X$ checks. 
    Note that the toric code's $X$ checks are deformed into weight-five checks which include qubits from the ancilla complex (orange edges). 
    (c) Decoding graph of $\partial_{\cone(f[i]),3}=M_Z$, where vertices are meta checks (volumes in (b)) and errors may occur on edges, which are $Z$ check measurements (faces in (b)).
    In particular, vertical edges are checks in $G_1\otimes D_0$ (blue faces in (b)) and horizontal edges are checks in $G_0\otimes D_1$ (gray faces in (b)).
    This decoding graph is a square lattice with one periodic boundary condition.
    The outlined blue edges are $Z$ checks whose product is the blue $\bar{Z}_1$ representative in (b).
    To flip the logical measurement and invalidate no meta checks, one needs a string error (red) that extends all the way around the decoding graph, and is thererfore a measurement error with weight $d$. 
    }
    \label{fig:toric-code}
\end{figure*}

This section presents an intuitive example of a fast surgery ancilla system on the toric code.
Although related to the constructions presented in section~\ref{sec:construction}, these gadgets are fault-tolerant even without any expansion on the base code or the ancilla complex maps (see Appendix~\ref{apdx:toric-code-proof}).
Further, we find that these gadgets give asymptotically constant space-time overhead per logical measurement.

In figure~\ref{fig:toric-code}~(a), we illustrate a base toric code, with data qubits on edges, $Z$ checks on faces, and $X$ checks are vertices. 
To connect this to our construction for general hypergraph product codes, let us also interpret this toric code as the hypergraph product of two cyclic repetition codes, $(C\otimes D)_\bullet$.
The space of qubits is $(C_1\otimes D_0)\oplus (C_0\otimes D_1)$, which we can interpret as the union of edges parallel to the meridional cycles and longitudinal cycles of the torus respectively.
For future reference, we highlight the edges in a minimum-weight representative of the logical $\bar{Z_1}$ in light blue, which we will measure. 
Additionally, this representative is completely in the support of qubits in $C_1\otimes D_0$.
We also highlight a representative of the $\bar{Z}_2$ logical in $(C_0\otimes D_1)$ in green solid lines, the $\bar{X_1}$ logical in $(C_1\otimes D_0)$ in orange with dotted lines, and the $\bar{X}_2$ logical in $(C_0\otimes D_1)$ in red with dotted lines.

From the lens of viewing lattice surgery as code deformation~\cite{vuillot2019code}, to measure a logical operator, we need to add this operator to the stabilizer group, such that measuring the stabilizers of the deformed code also measures the logical operator.
This means that the coned code, $\cone(f[i])$, illustrated in figure~\ref{fig:toric-code}(b), must include $\bar{Z}_1$ in the stabilizer group.
In~\ref{fig:toric-code}(b), once again qubits are on edges, $X$ checks are on vertices, and $Z$ checks are on faces. 
However in contrast to the toric code, there are additionally volumes, which we will see are meta checks on $Z$ checks.
The representative of $\bar{Z}_1$ that was highlighted in~\ref{fig:toric-code}(a) is also highlighted in light blue in~\ref{fig:toric-code}(b).
Illustrated to the right of the coned code is a slice of the filled torus that shows a set of blue $Z$ checks whose product is the $\bar{Z}_1$ representative, as the support of the checks on the orange edges cancel out.
This confirms that the logical operator $\bar{Z}_1$ can be inferred by the measurements of $Z$ checks. 
If we only examine this slice of the filled torus, this is exactly $\cone(g[i])$ (in equation~\eqref{eq:cone-Gi}), where the toric code qubits in blue are qubits in $C_1$, ancilla qubits in orange are in $G[i]_0$, and blue vertices are checks the $\cC$ repetition code, in $C_0$.
Finally, the highlighted blue $Z$ checks are in the space $G[i]_1$, whose support on the qubits in $\bar{Z}_1$ is given by the map $g[i]_1:G[i]_1\rightarrow C_1$.
For the particular $\cGi{i}$ in this example, there are no checks in $G[i]_{-1}$.

Suppose that $\cGi{i}$ was the only ancilla system we attached to the toric code and we measured checks in $G[i]_1$ for a single round.
Since the ancilla qubits in $ G[i]_0$ are initialized in $\ket{+}$, the measurement of each $G[i]_1$ check is random.
If a single one of these highlighted blue checks had a measurement flip, it is undetectable as the values of these checks are inherently random.
Worse, the inferred value of $\bar{Z}_1$ would also be flipped, resulting in the wrong logical measurement from a single measurement error.
How can we ensure that our logical operator measurement is resilient to measurement errors?
An intuitive solution is to make our measurement redundant by measuring other $\bar{Z}_1$ logical representatives.
This is exactly what we do by taking the product of $\cGi{i}$ with the repetition code $\cD$, and connecting it to the original toric code by the map $g[i]\otimes \id_D:(C\otimes D)_\bullet \rightarrow (G[i]\otimes D)_\bullet$. 
In~\ref{fig:toric-code}(b), one may interpret this construction as using copies of $\cGi{i}$ to `fill' the interior of each $\bar{Z}_1$.

For reference, we put the detailed chain complex of $\cone(f[i])=(\cone(g[i])\otimes D)_\bullet$ below to easily connect visual components of $\cone(f[i])$ in~\ref{fig:toric-code}(b) to spaces.

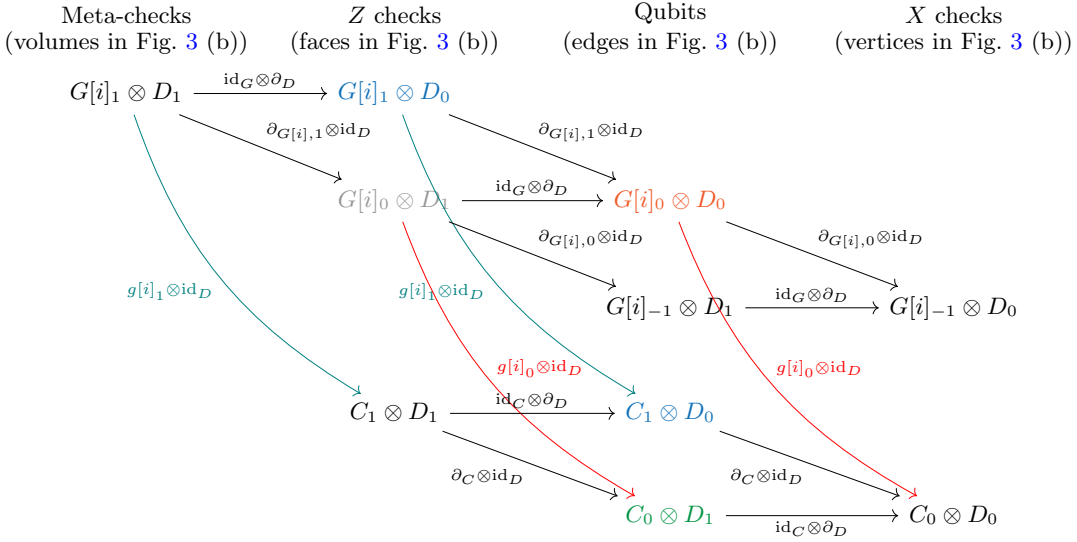
\begin{figure*}
    \begin{tikzcd}[column sep=5.5 em,row sep=large, /tikz/execute at end picture={
    \node (huge) [fit=(mz),label=above:{\shortstack{Meta-checks\\(volumes in Fig.~\ref{fig:toric-code}~(b))}}] {};
    \node (huge) [fit=(Z checks) ,label=above:{\shortstack{$Z$ checks\\(faces in Fig.~\ref{fig:toric-code}~(b))}}] {};
    \node (huge) [fit=(q), label=above:{\shortstack{Qubits\\(edges in Fig.~\ref{fig:toric-code}~(b))}}] {};
    \node (huge) [fit=(x),label=above:{\shortstack{$X$ checks\\(vertices in Fig.~\ref{fig:toric-code}~(b))}}] {};
  }]
  |[alias=mz]| G[i]_1\otimes D_1 \arrow[r, "\id_G\otimes \partial_D"] \arrow[rd, "\partial_{G[i],1}\otimes \id_D"] \arrow[to=c1d1, draw=teal, bend right=20, "\color{teal}g{[i]}_1 \otimes \id_D"']&
  |[alias=Z checks]| \textcolor{mylightblue}{G[i]_1\otimes D_0} \arrow[rd, "\partial_{G[i],1}\otimes \id_D"] \arrow[to=c1d0, bend right=20, draw=teal, "\color{teal}g{[i]}_1 \otimes \id_D"']&
  |[alias=q]| \phantom{G[i]_1\otimes D_0}&
  |[alias=x]| \phantom{G[i]_1\otimes D_0}&
  \\
  &
  \textcolor{mygray}{G[i]_0\otimes D_1} \arrow[r, "\id_G\otimes \partial_D"] \arrow[rd, "\partial_{G[i],0}\otimes \id_D"]\arrow[to=c0d1, bend right=20, draw=red, "\color{red}g{[i]}_0 \otimes \id_D"]&
  \textcolor{myorange}{G[i]_0\otimes D_0} \arrow[rd, "\partial_{G[i],0}\otimes \id_D"]\arrow[to=c0d0, bend right=20, draw=red, "\color{red}g{[i]}_0 \otimes \id_D"]
  \\
  &
  &
  G[i]_{-1}\otimes D_1 \arrow[r, "\id_G\otimes \partial_D"] &
  G[i]_{-1}\otimes D_0
  \\
  &
 |[alias=c1d1]|C_1\otimes D_1 \arrow[r, "\id_C \otimes \partial_{D}"] \arrow[rd, "\partial_C\otimes \id_D"'] &
 |[alias=c1d0]|\textcolor{mylightblue}{C_1\otimes D_0}  \arrow[rd, "\partial_C\otimes \id_D"'] &
 \\
 &
 &
 |[alias=c0d1]|\textcolor{mygreen}{C_0\otimes D_1} \arrow[r, "\id_C\otimes \partial_D"'] &
 |[alias=c0d0]|C_0\otimes D_0  &
\end{tikzcd}
\caption{The complete chain complex of a deformed code $\cone(f[i])=(\cone(g[i])\otimes D)_\bullet$. Spaces are color-coded to match the example in Fig.~\ref{fig:toric-code}. However, here we explicitly display all the spaces of a general coned code that are not shown in Fig.~\ref{fig:toric-code}, such as $G[i]_{-1}\otimes D_1$ and $G[i]_{-1}\otimes D_0$.}
\label{fig:explicit-toric-chain-map}
\end{figure*}
In~\ref{fig:toric-code}(b), we measure $\bar{Z}_1$ representatives through ancilla checks corresponding to faces whose product form a surface with a boundary corresponding to each $\bar{Z}_1$ representative.
These checks are in \textcolor{mylightblue}{$G[i]_1\otimes D_0$} in Fig.~\ref{fig:explicit-toric-chain-map}. 
Ancilla qubits are orange edges in~\ref{fig:toric-code}(b) and in \textcolor{myorange}{$G[i]_0\otimes D_0$} in Fig.~\eqref{fig:explicit-toric-chain-map}.
Additionally, we introduce horizontal $Z$ checks, shaded gray in the interior of the torus in~\ref{fig:toric-code}(b) and in \textcolor{mygray}{$G[i]_0\otimes D_1$} in Fig.~\ref{fig:explicit-toric-chain-map}.
The chain map of $\cone(f[i])$ shows that checks in \textcolor{mygray}{$G[i]_0\otimes D_1$} act on ancilla qubits in \textcolor{myorange}{$G[i]_0\otimes D_0$} as well as qubits of the original toric code in \textcolor{mygreen}{$C_0\otimes D_1$}.

Although checks in \textcolor{mygray}{$G[i]_0\otimes D_1$} are not strictly required to infer a measurement of $\bar{Z}_1$, the role of these checks is crucial in the following manner:
Suppose checks in \textcolor{mygray}{$G[i]_0\otimes D_1$} were not present, 
this is the same as measuring $d$~disjoint $\bar{Z}_1$ representatives in parallel using $d$ separate surgery ancillas.
The measurement outcome of each representative would be flipped with roughly the probability of a single measurement error occurring on any check in the product that equals $\bar{Z}_1$. 
As the weight of $\bar{Z}_1$ increases with growing toric code distance, so does the number of $Z$ checks \textcolor{mylightblue}{$G[i]_1\otimes D_0$} participating in the product representing $\bar{Z}_1$.
Consequently, as $d\rightarrow \infty$, the probability of a measurement flip on a $\bar{Z}_1$ measurement approaches $1$. 
Therefore, despite logical measurement errors only occurring after at least $ \lceil d/2\rceil $ measurement errors on separate representatives, the probability of a measurement error on a single representative approaches 1 asymptotically and thus overall, this measurement will only have a pseudo-threshold instead of a threshold.
Including $Z$ checks in \textcolor{mygray}{$G[i]_0\otimes D_1$}, induces constant-sized meta-checks, so that the check measurement errors are detected by meta-checks whose probability of flipping remains constant with growing code distance.

Meta-checks are volumes in~\ref{fig:toric-code}(b), and in $ G[i]_1\otimes D_1$ in Fig.~\ref{fig:explicit-toric-chain-map}.
As an example, take a measurement error of a single $Z$ check in \textcolor{mylightblue}{$G[i]_1\otimes D_0$}, denoted by a red `X'.
This results in the the flip of two meta-check volumes adjacent to the measurement error.
For this measurement error to go undetected, a string of measurement errors the length of the longitudinal cycle must occur around the torus.
The chain map of $\cone(f[i])$ formalizes how meta-checks check for measurement errors of checks in \textcolor{mylightblue}{$G[i]_1\otimes D_0$} $\oplus$ \textcolor{mygray}{$G[i]_0\otimes D_1$}, which is through the meta-check matrix $M_Z=(\id_G\otimes \partial_D)\oplus (\partial_{G[i],1}\otimes \id_D)$ where $\partial_D$ is simply another cyclic repetition code parity check matrix. 
This leads to the meta-check decoding graph in~\ref{fig:toric-code}(c), where meta-checks represented as vertices, and horizontal (vertical) edges are $Z$ checks in \textcolor{mylightblue}{$G[i]_1\otimes D_0$} (\textcolor{mygray}{$G[i]_0\otimes D_1$}). 
Undetected, logical measurement errors are codewords of $M_Z$, which are  closed loops that circumvent the ring in~\ref{fig:toric-code}(c).
An example is highlighted in red in~\ref{fig:toric-code}(c).
Technically, these meta-checks also have support on the original, face checks of the toric code, $C_1\otimes D_1$ (Fig.~\ref{fig:explicit-toric-chain-map}).
How would a decoder determine whether meta-check syndrome is due to measurement errors of \textcolor{mylightblue}{$G[i]_1\otimes D_0$} $\oplus$ \textcolor{mygray}{$G[i]_0\otimes D_1$} or from measurement errors of $C_1\otimes D_1$?
The solution is to use previous and future measurement outcomes of $C_1\otimes D_1$ to correct for measurement errors on $C_1\otimes D_1$ in $\cone(f[i])$.
Unlike the random outcomes of \textcolor{mylightblue}{$G[i]_1\otimes D_0$} $\oplus$ \textcolor{mygray}{$G[i]_0\otimes D_1$} checks, checks in $C_1\otimes D_1$ remain in the stabilizer group before and after switching to the coned code $\cone(f[i])$. 
As a result, their measurement outcomes may be used to construct detectors extending before and after the time-slice of $\cone(f[i])$, and we can use this to infer the location of measurement errors on $C_1\otimes D_1$. 
In summary, even though we only measure the stabilizers of $\cone(f[i])$ for a single round, $d$ rounds of syndrome history in the past and present are essential for allowing to correct for measurement errors on the $Z$ checks of the original code. 
The time-cost of these $d$ buffer rounds, however, can be amortized across
multiple logical measurements.

The space overhead of this ancilla system, is
\begin{align}
    &\dim(G[i]_0\otimes D_0) + \dim(G[i]_{-1}\otimes D_1)\notag\\
    &= \dim(G[i]_0)\dim(D_0)+ 0=O(d^2)
\end{align}
Thus the total space overhead of $\cone(f[i])$ is also $O(d^2)$, which is constant with respect to the space overhead of the base toric code, $O(d^2)$.

In Appendix~\ref{apdx:toric-code-proof}, we prove compacted code distance for these gadgets, which may also perform entangling measurements across copies of multiple toric codes. 
The lack of expansion of the map $\partial_{G[i],1}$ means that the proof for $d_1(\cc_\bullet)$ in proposition~\ref{prop:compacted-code-dist} must be modified. 
Fortunately, this means that we only have to prove the preservation of logical $Z$ distance in the compacted code.
A proof sketch is as follows.
We consider $d$, disjoint, and minimum-weight $\bar{X}_2$ logical representatives of the base toric code. 
For the general case, we impose a condition on the map $g[i]_1$ such that we can find an appropriate set of $d$ deformed and disjoint representatives for every $X$ logical conjugate to an unmeasured $Z$ logical.
In the example in Fig.~\ref{fig:toric-code}, we only need to consider $\bar{X}_2$.
Visually, we can verify that every new $Z$ check in $\cone(f[i]$) commutes with these particular $\bar{X}_2$ logical representatives, so they remain logical representatives in $\cone(f[i]$).
Since there are $d$, disjoint $\bar{X}_2$ representatives, we require that $\bar{Z}_2$, the conjugate logical $Z$, must have weight at least $d$ to anticommute with each $\bar{X}_2$ representative.
In Appendix~\ref{apdx:toric-code-proof}, we formally extend this argument to more unmeasured $Z$ logical operators on multiple toric code copies for logical measurements which may entangle multiple toric codes.

\section{Future directions}
\label{sec:future}

Our work demonstrates that careful construction of ancillary complexes enables addressable and parallel logical measurements in nearly constant spacetime for hypergraph product codes.
Our gadgets measure logical operators along rows or columns in parallel by taking a hypergraph product of a measurement graph with one of the original classical codes that define the base code.
Perhaps one can achieve further addressability within a given row or column of logical qubits by augmenting the original classical code before taking the product through careful puncturing~\cite{xu2024fast}. 
The challenge would be to prove that after such augmentation, the distance of the compacted code remains at least $d$.
We remark that for toric codes, each code block has few enough logical qubits such that one can perform CSS measurements that are arbitrarily addressable on any set of logical qubits.
When these measurements are interleaved with fold-transversal gates of the toric code, any Clifford circuit with depth at least $d$ could be executed with constant space-time overhead.
This would be an interesting example of extending algorithmic fault-tolerance~\cite{zhou2024algorithmic}, which has only been shown for surface codes with transversal gates, to lattice surgery on code blocks with more than one logical qubit.
We leave the proof of distance preservation of the toric fast surgery gadgets with intermittent transversal gates for future work.
A natural next direction is to find fast and low-overhead surgery gadgets for other codes such as bicycle bivariate codes and product codes with better distance scaling, namely lifted and balanced product codes.

\section*{Acknowledgements}
We thank Louis Golowich, Alexander Cowtan, John Blue and Qian Xu for insightful discussions.
Z.H. acknowledges support from the MIT Department of Mathematics, the MIT-IBM Watson AI Lab, and the NSF Graduate Research Fellowship Program under Grant No. 2141064.
K.C. acknowledges Alfred.

\bibliography{bib}

\clearpage
\onecolumngrid
\appendix

\section{Omitted Proofs From Constant-Time Surgery on HGP Codes}\label{apdx:proofs}

\coneproduct*
\begin{proof}
    From the definition of homological product of chain complexes and mapping cones, we have
    \begin{align}
         \cone(f\otimes \id_{D})_k
        &= (A\otimes D)_{k-1}\oplus (C\otimes D)_k \\
        &= \left( \bigoplus_{i+j = k-1}A_i\otimes D_j \right) \oplus \left( \bigoplus_{l+m = k} C_l\otimes D_m \right). \\
         \left(\cone(f)\otimes D\right)_k 
        &= \bigoplus_{p+q = k} \cone(f)_p\otimes D_q \\
        &= \bigoplus_{p+q = k} (A_{p-1}\oplus C_p)\otimes D_q \\
        &= \bigoplus_{p+q = k} (A_{p-1}\otimes D_q)\oplus (C_p\otimes D_q).
    \end{align}
    We see that the graded vector spaces of these two complexes are isomorphic.
    To check their boundary maps, take $p,q$ such that $p+q = k$. For element $a\otimes d\in A_p\otimes D_q$, we see that
    \begin{align}
        \partial_{\cone(f\otimes \id_{D}), k+1}(a\otimes d) 
        &= f_p\otimes \id_{D_q}(a\otimes d) - \partial_{A\otimes D,k}(a\otimes d) \\
        &= f_p(a)\otimes d  - \partial_{A,p}(a)\otimes d - (-1)^p a\otimes \partial_{D,q}(d). \\
         \partial_{\cone(f)\otimes D, k+1}(a\otimes d) 
        &= \partial_{\cone(f),p+1}(a)\otimes d + (-1)^{p+1}a\otimes \partial_{D, q}(d) \\
        &= (f_p(a) - \partial_{A, p}(a))\otimes d + (-1)^{p+1}a\otimes \partial_{D, q}(d).
    \end{align}
    Note that in evaluating $\partial_{\cone(f\otimes \id_{D}), k+1}(a\otimes d)$, we considered the action of chain map $f\otimes \id_{D}$ and of the boundary map of $(A\otimes D)_\bullet$, then expanded the second term. 
    In evaluating $\partial_{\cone(f)\otimes D, k+1}(a\otimes d)$, we first expanded the boundary map of homological product, then expanded the boundary map from a mapping cone. 
    We see that the two evaluations are equal.
    
    Similarly, for element $c\otimes d\in C_p\otimes D_q$, we have
    \begin{align}
        \partial_{\cone(f\otimes \id_{D}), k}(c\otimes d)
        = \partial_{C\otimes D, k}(c\otimes d) 
    \end{align}
    and 
    \begin{align}
         \partial_{\cone(f)\otimes D, k}(c\otimes d) 
        &= \partial_{\cone(f),p}(c)\otimes d + (-1)^{p}c\otimes \partial_{D, q}(d) \\
        &= \partial_{C,p}(c)\otimes d + (-1)^{p}c\otimes \partial_{D, q}(d) \\
        &= \partial_{C\otimes D, k}(c\otimes d).
    \end{align}
    We therefore conclude that the two chain complexes are isomorphic, where the isomorphism is simply the identity map between their graded vector spaces. 
\end{proof}

\CompactedConeDistance*
\begin{proof}
For easier reference, we redraw the coned complex as in equation~\eqref{eq:cone-G}.
\begin{equation}
\begin{tikzcd}[column sep=large,row sep=large]
|[alias=A1]| \oplus_{i\in [t]} G[i]_1 
\arrow[r, "\partial_{G,1}"] 
\arrow[rd, "\bar{g}_1"'] &
|[alias=A0]| \oplus_{i\in [t]} G[i]_0 
\arrow[r, "\partial_{G,0}"] 
\arrow[rd, "\bar{g}_0"'] &
|[alias=Am1]| \oplus_{i\in [t]} G[i]_{-1} \\
|[alias=C2] |0 \arrow[r, "0"'] &
|[alias=C1]| C_1 \arrow[r, "\partial_{C}"'] &
|[alias=C0]| C_0
\end{tikzcd}
\end{equation}
The fact that $d^1(\cone(\bar{g})) \ge 1$ trivially holds, therefore we focus on analyzing $d_1(\cone(\bar{g}))$. 
From Proposition~\ref{prop:1D-cone-compacted-homology}, we know that the elements in $H_1(\cone(\bar{g}))$ are precisely the elements in $H_1(\cC)$ which are not measured by the ancilla systems $\cGi{i}$. 
Let $c_{t+1}, \cdots, c_{k_C}$ be a linearly independent basis of $\ker(\partial_C)/\Span\{\{c_1\},\cdots, \{c_t\}\}$.
Each vector $c_j$ for $t+1\le j\le k_C$ corresponds to a chain $(c_j,0,\cdots, 0)\in H_1(\cone(\bar{g}))$, where the $0$s are in the spaces $G[i]_0$.
It suffices for us to analyze how these chains deform given the newly added space $\oplus_{i\in [t]} G[i]_1$.

Fix such a chain $(c_j,0,\cdots, 0)$, consider an arbitrary element $(v_1, \cdots, v_t)\in \oplus_{i\in [t]} G[i]_1$ where $v_i \in G[i]_1$. 
We have 
\begin{equation}
\begin{split}
    \partial_{G, 1}((v_1, \cdots, v_t)) + (c_j,0,\cdots, 0) 
    &= (c_j + \sum_{i\in [t]}g[i]_1(v_i), \partial_{G[1],1}(v_1), \cdots, \partial_{G[t],1}(v_t)).
\end{split}
\end{equation}
From Lemma~\ref{lem:graph-surgery}, we know that $\partial_{G[i],1}$ has boundary Cheeger constant at least $1$. 
This implies that for all $i\in [t]$, 
\begin{align}
    |\partial_{G[i],1}(v_i)|\ge \min(|v_i|, \dim(G[i]_1) - |v_i|).
\end{align}
Since the maps $g[i]_1$ are $1$-sparse, and $g[i]_1$ applied to the all $1$s vector on $G[i]_1$ gives $c_i$, we see that $\dim(G[i]_1)\ge |c_i|$ and $|v_i|\ge g[i]_1(v_i)$.
It suffices for us to show that 
\begin{equation}
\begin{split}
\left|(c_j + \sum_{i\in [t]}g[i]_1(v_i), \partial_{G[1],1}(v_1), \cdots, \partial_{G[i],1}(v_i)) \right| 
&= \left| c_j + \sum_{i\in [t]}g[i]_1(v_i) \right| + \sum_{i\in [t]} |\partial_{G[i],1}(v_i)|\ge d. 
\end{split}
\label{eq:app-dist-inequality}
\end{equation}

\noindent Without loss of generality, suppose for $i\le \ell$ we have $|\partial_{G[i],1}(v_i)| = |v_i|$ and for $\ell+1\le i\le t$ we have $|\partial_{G[i],1}(v_i)| = \dim(G[i]_1) - |v_i|$.
By the triangle inequality, we have 
\begin{align}
\left| c_j + \sum_{i\in [t]}g[i]_1(v_i) \right| + \sum_{i\in [t]} |\partial_{G[i],1}(v_i)| 
&\ge \left| c_j + \sum_{i > \ell}g[i]_1(v_i) \right| - \sum_{i\le \ell}|g[i]_1(v_i)| + \sum_{i\in [t]} |\partial_{G[i],1}(v_i)|, \\
&\ge \left| c_j + \sum_{i > \ell}g[i]_1(v_i) \right| + \sum_{i > \ell} |\partial_{G[i],1}(v_i)|, \\
&\ge \left| c_j + \sum_{i > \ell}g[i]_1(v_i) \right| + \sum_{i > \ell} (\dim(G[i]_1) - |v_i|).
\label{eq:dist-inequality-1}
\end{align}
Note that the support of $g[i]_1(v_i)$ is strictly contained in the support of $c_i$. 
Let us write
\begin{align}
    c_j + \sum_{i > \ell}g[i]_1(v_i)
    &= c_j + \sum_{i > \ell} (c_i + (c_i - g[i]_1(v_i))).
\end{align}
Note here that the signs are flexible since we are working over $\bF_2$.
Let $\tilde{c} = c_j + \sum_{i > \ell} c_i$. We see that $\tilde{c}\in \ker(\partial_C)$, which means $|\tilde{c}|\ge d$. In other words, $\tilde{c}$ is another logical operator of the code $\cC$ and in $C_1$, so it therefore must have weight at least $q d$.
By the reverse triangle inequality, we have 
\begin{align}
    \left| c_j + \sum_{i > \ell}g[i]_1(v_i) \right|
    &=\left| \tilde{c} - \sum_{i>l}(c_i - g[i]_1(v_i))\right|\ge |\tilde{c}| - \sum_{i > \ell} |c_i - g[i]_1(v_i)|.
    \label{eq:dist-inequality-2}
\end{align}
Let $\mathbb{1}_i$ denote the all ones vector on $G[i]_1$, and note that $|\mathbb{1}_i - v_i| = \dim(G[i]_1) - |v_i|$. Moreover, $g[i]_1(\mathbb{1}_i - v_i) = c_i - g[i]_1(v_i)$. Since $g[i]_1$ is $1$-sparse, we have 
\begin{align}
    |c_i - g[i]_1(v_i)| &= |g[i]_1(\mathbb{1}_i - v_i)| 
    \le |\mathbb{1}_i - v_i| 
    =\dim(G[i]_1) - |v_i|
    \label{eq:dist-inequality-3}
\end{align}
Revisiting the left-hand side of the inequality in Eq.~\eqref{eq:app-dist-inequality}, and using Eqs.~\eqref{eq:dist-inequality-1},~\eqref{eq:dist-inequality-2}, and~\eqref{eq:dist-inequality-3}, we have
\begin{align}
    \left| c_j + \sum_{i\in [t]}g[i]_1(v_i) \right| + \sum_{i\in [t]} |\partial_{G[i],1}(v_i)|
    &\ge |\tilde{c}| - \sum_{i > \ell} (\dim(G[i]_1) - |v_i|)  + \sum_{i > \ell} (\dim(G[i]_1) - |v_i|)\ge|\tilde{c}|\ge d.
\end{align}
This proves our claim that $d_1(\cone(\bar{g})) \ge d$.
\end{proof}

\subsection{Relaxation to Relative Expansion}

We now relax our proof of Proposition~\ref{prop:1D-cone-compacted-distance} to using a weaker notion of expansion, called relative expansion, introduced in Ref.~\cite{swaroop2024universal}.

\begin{definition}[Relative Expansion]\label{def:relative_exp}
    Consider a matrix $M: \bF_2^n\rightarrow \bF_2^m$ such that the all ones vector $v$ is in the kernel of $M$. 
    Let $P\subseteq [n]$ be a subset of indices in $\bF_2^n$. 
    For a vector $v\in \bF_2^n$, let $S(v)$ denote the set of indices that $v$ has ones on.
    Let $t$ be an integer. 
    We define the relative Cheeger constant $\beta_t(M,P)$ to be the maximum value such that for all $v\in \bF_2^n$, we have
    \begin{align}
        |Mv| \ge \beta_t(M,P)\cdot \min(t, |S(v)\cap P|, |P\setminus S(v)|).
    \end{align}
\end{definition}

\begin{proposition}\label{prop:dist_relative_exp}
    Suppose the ancillary complexes $G[i]$ satisfy
    $\beta_d(G[i]_i,P_i) \ge 1$ for all $i$, where $P_i$ corresponds to the indices in $G[i]_1$ which are $1$-to-$1$ connected to $c_i$ by $\partial_{G[i],1}$. In other words,
    \[
    |\partial_{G[i],1}(v_i)| \ge  \min(d, |S(v_i)\cap P_i|, |P_i\setminus S(v_i)|).
    \]
    Then $d_1(\cone(\bar{g})\ge d$. 
\end{proposition}
\begin{proof}
We follow the same proof as in Proposition~\ref{prop:1D-cone-compacted-distance}, and again refer to the complex in equation~\eqref{eq:cone-G}.
Let $c_{t+1}, \cdots, c_{k_C}$ be a linearly independent basis of $\ker(\partial_C)/\Span\{\{c_1\},\cdots, \{c_t\}\}$.
Each vector $c_j$ for $t+1\le j\le k_C$ corresponds to a chain $(c_j,0,\cdots, 0)\in H_1(\cone(\bar{g}))$, where the $0$s are in the spaces $G[i]_0$.
It suffices for us to analyze how these chains deform given the newly added space $\oplus_{i\in [t]} G[i]_1$.

Fix a chain $(c_j,0,\cdots, 0)$, consider an arbitrary element $(v_1, \cdots, v_t)\in \oplus_{i\in [t]} G[i]_1$ where $v_i \in G[i]_1$. 
We have 
\begin{equation}
\begin{split}
    \partial_{G, 1}((v_1, \cdots, v_t)) + (c_j,0,\cdots, 0)
    &= (c_j + \sum_{i\in [t]}g[i]_1(v_i), \partial_{G[1],1}(v_1), \cdots, \partial_{G[t],1}(v_t)).
\end{split}
\end{equation}
It suffices for us to show that 
\begin{equation}
\begin{split}
\left|(c_j + \sum_{i\in [t]}g[i]_1(v_i), \partial_{G[1],1}(v_1), \cdots, \partial_{G[i],1}(v_i)) \right| 
&= \left| c_j + \sum_{i\in [t]}g[i]_1(v_i) \right| + \sum_{i\in [t]} |\partial_{G[i],1}(v_i)|\ge d. 
\end{split}
\label{eq:app-dist-inequality-relative-expansion-prop}
\end{equation}

\noindent Given the relative expansion properties of $\partial_{G[i],1}$, we have for all $i\in [t]$,
\begin{align}
    |\partial_{G[i],1}(v_i)| \ge  \min(d, |S(v_i)\cap P_i|, |P_i\setminus S(v_i)|).
\end{align}

\noindent If there is $i$ such that $|\partial_{G[i],1}(v_i)|\ge d$, then equation~\eqref{eq:app-dist-inequality-relative-expansion-prop} holds. 
Therefore, without loss of generality, suppose for $i\le \ell$ we have 
\begin{align}
    |\partial_{G[i],1}(v_i)| = |S(v_i)\cap P_i| = |g[i]_1(v_i)|,
\end{align}
and for $\ell+1\le i\le t$ we have 
\begin{align}
    |\partial_{G[i],1}(v_i)| = |P_i\setminus S(v_i)| = |c_i| - |g[i]_1(v_i)|.
\end{align}
By the triangle inequality, we have 

\begin{align}
\left| c_j + \sum_{i\in [t]}g[i]_1(v_i) \right| + \sum_{i\in [t]} |\partial_{G[i],1}(v_i)| 
&\ge \left| c_j + \sum_{i > \ell}g[i]_1(v_i) \right| - \sum_{i\le \ell}|g[i]_1(v_i)| + \sum_{i\in [t]} |\partial_{G[i],1}(v_i)|, \\
&\ge \left| c_j + \sum_{i > \ell}g[i]_1(v_i) \right| + \sum_{i > \ell} |\partial_{G[i],1}(v_i)|, \\
&\ge \left| c_j + \sum_{i > \ell}g[i]_1(v_i) \right| + \sum_{i > \ell} (|c_i| - |g[i]_1(v_i)|).
\label{eq:dist-inequality-1-relative-expansion-prop}
\end{align}

Let us write
\begin{align}
    c_j + \sum_{i > \ell}g[i]_1(v_i)
    &= c_j + \sum_{i > \ell} (c_i + (c_i - g[i]_1(v_i))).
\end{align}
Note here that the signs are flexible since we are working over $\bF_2$.
Let $\tilde{c} = c_j + \sum_{i > \ell} c_i$. We see that $\tilde{c}\in \ker(\partial_C)$, which means $|\tilde{c}|\ge d$. In other words, $\tilde{c}$ is another logical operator of the code $\cC$ and in $C_1$, so it therefore must have weight at least $q d$.
By the reverse triangle inequality, we have 
\begin{align}
    \left| c_j + \sum_{i > \ell}g[i]_1(v_i) \right|
    &=\left| \tilde{c} - \sum_{i>l}(c_i - g[i]_1(v_i))\right|\ge |\tilde{c}| - \sum_{i > \ell} |c_i - g[i]_1(v_i)|.
\end{align}

Combining this with equation~\eqref{eq:dist-inequality-1-relative-expansion-prop}, we have
\begin{align}
    \left| c_j + \sum_{i\in [t]}g[i]_1(v_i) \right| + \sum_{i\in [t]} |\partial_{G[i],1}(v_i)|
    &\ge |\tilde{c}| - \sum_{i > \ell} (\dim(G[i]_1) - |v_i|)  + \sum_{i > \ell} (\dim(G[i]_1) - |v_i|)\ge|\tilde{c}|\ge d.
\end{align}

This proves our claim that $d_1(\cone(\bar{g})) \ge d$.
\end{proof}

\section{Constant spacetime overhead surgery for the toric code}
\label{apdx:toric-code-proof}

The toric code gadgets we present follow the same construction in the main text, except the $\cG$ complex we use differs from Lemma~\ref{lem:graph-surgery} in that condition~\ref{surgery-cond:expansion} is relaxed.
Since $\partial_{G[i],1}$ does not have a boundary Cheeger constant of at least 1, we must use a different approach than the proof of proposition~\ref{prop:1D-cone-compacted-distance}, which relies on expansion of $\partial_{G[i],1}$ to show $d_1(\cone(\bar{g}))\ge d$ and therefore $d_1(\cc_\bullet)\ge d$.
Previously, the proof of proposition~\ref{prop:1D-cone-compacted-distance} considered the one-dimensional cone $\cone(g[i])$.
Instead, we will jump straight to proving the $d_1(\cone(f[i]))\ge d$, then iteratively construct the compacted code $\cc_\bullet = \cone(f)=\bigoplus_{i\in [t]}\cone(f[i])$.

Toric codes are hypergraph product codes where $\cC$ and $\cD$ are cyclic repetition codes. 
Let us assume that $n_C=n_D=d_C = d_D =m_C=m_D= d$. 
This similarly holds for the dual codes, $C^\bullet:C^0\xrightarrow[]{\delta^C}C^1$ and $D^\bullet:D^0\xrightarrow[]{\delta^D}D^1$, where $\partial_C=\delta^C$ and $\partial_D=\delta^D$.
For reference, the chain complex of the toric code is:

\begin{equation}
\begin{tikzcd}[column sep=large,row sep=large, /tikz/execute at end picture={
    \node (huge) [fit=(z),label=above:{Z checks}] {};
    \node (huge) [fit=(qv) ,label=above:{$q^\mathrm{v}$}] {};
    \node (huge) [fit=(qh), label=below:{$q^\mathrm{h}$}] {};
    \node (huge) [fit=(x),label=above:{X checks}] {};
  }]
|[alias=z]| C_1\otimes D_1 \arrow[r, "\id_C \otimes \partial_{D}"] \arrow[rd, "\partial_C\otimes \id_D"'] &
|[alias=qv]| C_1\otimes D_0   \arrow[rd, "\partial_C\otimes \id_D"'] &
\\
&
|[alias=qh] |C_0\times D_1 \arrow[r, "\id_C\otimes \partial_D"'] &
|[alias=x]| C_1 
\end{tikzcd}
\end{equation}
Let $(q^\mathrm{v}, q^\mathrm{h})\in (\mathbb{F}_2^{n_Cm_D},\mathbb{F}_2^{m_Cn_D})$ be a vector that denotes the support of qubits in the two spaces $C_1\otimes D_0$ and $C_0 \otimes D_1$. 
The logical operators of the toric code can be written as
\begin{align}
    &\{(c \otimes e_l, 0), (0,e_h\otimes d)\} \in H_1((C\otimes D)_\bullet),\label{eq:canonical-toric-zlogical}\\
    &\{(e_g\otimes \bar{d},0),(0,\bar{c}\otimes e_k)\}\in H^1((C\otimes D)^\bullet),
    \label{eq:canonical-toric-xlogical}
\end{align}
where $c,e_g\in \mathbb{F}_2^{n_c}$, $\bar{d},e_l\in \mathbb{F}_2^{m_d}$, $\bar{c},e_h\in \mathbb{F}_2^{m_c}$, and $d,e_k\in \mathbb{F}_2^{n_d}$.
The vectors $c,d,\bar{c},$ and $
\bar{d}$ are in $\ker(\partial_C)$, $\ker(\partial_D)$, $\ker(\delta^C)$, and $\ker(\delta^D)$ respectively. 
Specifically, $c=d=\bar{c}=\bar{d}=\mathbb{1}_d$, or the all ones column vector in $\mathbb{F}_2^d$.
Let $e_g$ be the $g$th unit vector in $ \mathbb{F}_2^{n_c}$, $e_k$ be the $k$th unit vector in $\mathbb{F}_2^{n_d}$, $e_h$ the $h$th unit vector in $\mathbb{F}_2^{m_c}$, and $e_l$ the $l$th unit vector in $\mathbb{F}_2^{m_d}$.

We are interested in the distance of the compacted code for $t$ surgery operations, which will contain measurements over multiple toric codes.
Take $M$ toric codes that the $t$ surgery operations will act on,
\begin{align}
    \bigoplus_{j\in [M]}(\cC\otimes \cD)^{(j)} 
    = (\bigoplus_{j\in [M]}\cC^{(j)})\otimes \cD 
    = \cC\otimes (\bigoplus_{j\in [M]}\cD^{(j)}).
\end{align} 
and let $\cC' = C_1'\xrightarrow[]{\partial_C'} C_0'$, where
\begin{align}
    C_i'=\bigoplus_{j\in [M]} C_i ^{(j)},\quad \quad \partial_C' = \bigoplus_{j\in [M]} \partial_{C}^{(j)}=I_M\otimes \partial_C.
\end{align}
The last equality for $\partial_C'$ uses that fact that $\partial_C^{(j)}=\partial_C^{(j')}$ for all $j\ne j'$.
$I_M$ is the $M\times M$ Identity matrix.
Since $(C'\otimes D)_\bullet$ are just many copies of the toric code, we can write the logical operators as
\begin{align}
    &\{(c'_m \otimes e_l, 0), (0,e_H\otimes d)\} 
    \in H_1((C\otimes D)_\bullet),
    \label{eq:extended-toric-zlogical}\\
    &\{(e_G\otimes \bar{d},0),
    (0,\bar{c}_m'\otimes e_k)\}
    \in H^1((C\otimes D)^\bullet).
    \label{eq:extended-toric-xlogical}
\end{align}
where $\{e_H\}_{H=1}^{M\cdot m_c}$ is the unit vector basis for $\mathbb{F}_2^{m\cdot m_C}$, which we rewrite was the tensor product of two unit vectors, $\{e_m\otimes e_h\}$. 
Here, $\{e_m\}_{m=1}^M$ is the unit vector basis of $\mathbb{F}_2^{M}$.
Similarly, $\{e_G\}_{G=1}^{m\cdot n_C}$ is the unit vector basis for $\mathbb{F}_2^{m\cdot n_C}$, which we also rewrite as $\{e_m\otimes e_g\}$.
Let $\{b_m\}$ for $m\in [M]$ be a different choice of basis for $\mathbb{F}_2^M$.
We can show that $\{c'_m\}_{m=1}^M$ and $\{\bar c'\}_{m=1}^M$, which are bases for $\ker(\partial'_C)$ and $\ker(\delta'^C)$, can be chosen as
\begin{align}
    c'_m = b_m\otimes c,\quad\quad \bar{c}'_m=b_m\otimes \bar{c}.
\end{align}
We walk through this argument for $\ker(\partial_C')$, since the same argument applies for $\ker(\delta'^C)$.
\begin{lemma}
    Let $\{b_m\}$ be a basis for $\mathbb{F}_2^M$. 
    Let $\partial'_C=I_M\otimes \partial_C$.
    Let $c\in \ker(\partial_C)$.
    Then,
    \[
    \ker(\partial'_C)=\mathrm{Span}\{b_m\otimes c\}.
    \]
\end{lemma}
\begin{proof}
    Take any $b_m\otimes c$ where $c\in \ker(\partial_C)$.
    We see that
    \begin{align}
        (I_M\otimes \partial_C)(b_m\otimes c)= b_m\otimes \partial_C c = 0,
    \end{align}
    so $b_m\otimes c\in \ker(\partial_C')$ for all $m\in [M]$.
    Now, take any element of $x\in \mathbb{F}^M_2\otimes C_1$,
    \begin{align}
        x=\sum_{m=1}^{M} b_m\otimes c_m
    \end{align}
    for some vectors $c_m\in C_1$.
    If $x$ is in the kernel of $\partial_C'$, then that means
    \begin{align}
        (I_M\otimes \partial_C)x = \sum_{m=1}^M b_m\otimes \partial_C c_m = 0
    \end{align}
    Since $\{b_m\}$ is a basis, then this means that $\ \partial_Cc_m=0$ for all $m$.
    Thus, $x\in \mathbb{F}_2^M\otimes\ker(\partial_C)$.
    Since $\{c\}$ is the basis for $\ker(\partial_C)$, then $\{b_m\otimes c\}$ is spanning and linearly independent.
\end{proof}
We will consider constant time measurements of products of $Z$ logicals of these $M$ toric codes. For example, take the $Z$ logical $((b\otimes c)\otimes e_l,0)$, where $b\in \mathbb{F}_2^{M}$. 
This operator is the product of the $Z$ logicals $((e_m\otimes c)\otimes e_l,0)$ over all indices $m$ such that $b_m=1$.
It is convenient to switch to a basis for $\ker(\partial'_C)$ that contains $b$ as a basis vector instead of $\{e_m\otimes c\}$.
This choice makes it straightforward to describe the remaining $Z$ logicals of the deformed code.
These logicals include $((b^\perp_i\otimes c)\otimes  e_l,0)$, where $\{b^\perp_i\}$ forms a basis for $\mathbb{F}_2^M/\ \mathrm{Im}(b)$.
In Lemma~\ref{lem:orthogonal-basis}, we show a useful commutation relation using this new basis for logical operators. 

\begin{lemma}\label{lem:orthogonal-basis}
    Let $B=\{b_1,\dots,b_t\}\subset \mathbb{F}_2^M$ be a set of linearly independent vectors.
    Define the orthogonal subspace 
    \begin{align}
        B^\perp:=\{x\in \mathbb{F}_2^M: x^\top b_i=0 \text{ for all } i\in [t]\}.
    \end{align}
    Then, $B^\perp$ is a subspace of dimension $M-t$, and therefore admits a basis $\{b_i^\perp\}$ of $M-t$ linearly independent vectors.
    It follows that logical operators
    $\bar{Z}_{b_j}^{\mathrm{v}}=((b_j\otimes c)\otimes e_l,0)$ for $b_j\in B$ commutes with every $X$ logical that has the form $\bar{X}^{\mathrm{v}}_i=((b^\perp_i\otimes e_g)\otimes \bar{d},0)$ for $b_i^\perp\in B^\perp$.
\end{lemma}

\begin{proof}
    The subspace $B^\perp$ is defined by $t$ independent linear constraints on vectors in $\mathbb{F}_2^M$, so $\dim(B^\perp)=M-t$.
    Thus, this space must have a basis of $M-t$ linearly independent vectors.
    To get the last result, we can straightforwardly see the $\bar{Z}^{\mathrm{v}}_b$ and $\bar{X}_i^{\mathrm{v}}$ operators commute,
    \begin{align}
        {(\bar{X}^\mathrm{v}_i)}^\top  \bar{Z}_{b_j}^\mathrm{v}=((({b_i^\perp})^\top \otimes e_g^\top )\otimes \bar{d}^\top )( ({b_j}\otimes c)\otimes e_l)  + 0 = {b_i^\perp} ^\top  b_j\otimes e_g^\top  c\otimes \bar{d}^\top e_l=0,
    \end{align}
    by using $({b_i^\perp})^\top  b_j=0$ for all $i,j$ since $b_i^\perp\in B^\perp$.  
\end{proof}
Without loss of generality, suppose that we measure a $Z$ logical operator $\bar{Z}_b^\mathrm{v} = (c' \otimes e_l, 0)$ for $c'=(b\otimes c)\in \ker(\partial'_C)$.
Here, $b\in \mathbb{F}_2^{M}$ denotes which of the $M$ toric codes the operator $\bar{Z}^\mathrm{v}_g $ acts non-trivially on.
The deformed code that mediates the measurement is the coned code $\cone(g[i])$,
\begin{equation}
\begin{tikzcd}[column sep=large,row sep=large]
    |[alias=A1]| G[i]_1 
    \arrow[r, "\partial_{G,1}"] 
    \arrow[rd, "g{[i]}_1"'] &
    |[alias=A0]| G[i]_0 
    \arrow[r, "\partial_{G,0}"] 
    \arrow[rd, "g{[i]}_0"'] &
    |[alias=Am1]| G[i]_{-1} \\
    |[alias=C2] |0 \arrow[r, "0"'] &
    |[alias=C1]| C'_1 \arrow[r, "\partial'_{C}"'] &
    |[alias=C0]| C'_0.
\end{tikzcd}
\end{equation}
Here $g[i]_\ell:\mathbb{F}_2^{|G[i]_\ell|}\rightarrow \mathbb{F}_2^{M\cdot |C_\ell|}$.
We additionally impose the condition that 
\begin{align}
    &g[i]_1^{-1}(b_i^\perp\otimes  e_g)=0, \quad b_i^\perp\in B^\perp\subset  \mathbb{F}_2^{M},\ e_g \in\mathbb{F}_2^{n_C}\label{eq:g1-condition}
\end{align}
where $b_i^\perp$ is a basis vector of the orthogonal complement of $b$, and $e_g$ is the $g$th unit vector of $\mathbb{F}_2^{n_c}$.
Let $\cone(f[i])=(\cone(g[i])\otimes D)_\bullet$. 
Note that relaxing condition~\ref{surgery-cond:expansion} of Lemma~\ref{lem:graph-surgery} and imposing equation~\eqref{eq:g1-condition} does not affect the proof or statements of Lemma~\ref{lem:cone_product}, proposition~\ref{prop:coned-homology} or proposition~\ref{prop:compacted-cone-0-codistance}.
Thus, the logical action and $d^1(\cc_\bullet)$ of these toric code gadgets is unchanged from section~\ref{sec:construction}.

\begin{remark}
    $Z$ logicals of the original toric code, $\bar{Z}^{\mathrm{h}}_m=(0,(e_m\otimes e_h)\otimes d)$ and $\bar{Z}^\mathrm{v}_i=((b_i^\perp\otimes c)\otimes e_l,0)$ are undeformed when switching to the deformed code because there are no new, ancilla $X$ checks that have support on $C_1\otimes D_0$ or $C_0\otimes D_1$ (Fig.~\ref{fig:explicit-toric-chain-map}).
    Therefore, the following trivial extensions of $\bar{Z}^{\mathrm{h}}_m$ and $\bar{Z}^\mathrm{v}_i$ to the full ancilla and code system gives homologically equivalent coned code logicals,
    \begin{align}
        \bar{Z}^{\mathrm{h},\cone}_m&=(0,0,0,(e_m\otimes e_h)\otimes d) \in (G[i]_0\otimes D_0,G[i]
        _{-1}\otimes D_0, C'_1\otimes D_0,C'_0\otimes D_1) \\
        \bar{Z}^\mathrm{v,\cone}_i&=(0,0,(b_i^\perp\otimes c)\otimes  e_l,0)\in (G[i]_0\otimes D_0,G[i]
        _{-1}\otimes D_0, C'_1\otimes D_0,C'_0\otimes D_1)
    \end{align}
    \label{rmk:undeformed-z-cone}
\end{remark}

\begin{proposition}
    Logicals $\bar{Z}^{\mathrm{h},\cone}_m=(0,0,0,(e_m\otimes e_h)\otimes d)\in H_1(\cone(f[i]))$ must have distance at least $d$.
    \label{prop:zh-distance}
\end{proposition}
\begin{proof}
    Take $\bar{X}^\mathrm{h}_m= (0,\bar{c}'_m\otimes e_k)$ logical operators of the original toric code(s) where $\bar{c}'_m\in \ker(\delta'^C)$ and $e_k$ is the $k$th unit vector of $\mathbb{F}_2^{n_d}$. 
    We claim that the following operator on the ancilla and code qubit vector space is an $X$ logical operator of the coned code,
    \begin{align}
        &\bar{X}^\mathrm{h,cone}_m=(0,\partial_{G,0}(g_0^{-1}[i](\bar{c}_m'))\otimes e_k,0,\bar{c}'_m\otimes e_k)\notag\\
        &\in (G[i]_0\otimes D_0,G[i]_{-1}\otimes D_1,C'_1\otimes D_0, C'_0\otimes D_1).
        \label{eq:xh-disjoint-logicals}
    \end{align}
    We can verify $\bar{X}^\mathrm{h,cone}_m$ is indeed in $ H^1(\cone(f[i]))$ because it is in $ \ker(\delta^{\cone(f[i]),1})$, where
        
    \begin{align}
        \delta^{\cone(f[i]),1} = \begin{pNiceMatrix}[first-row, first-col]
        & G[i]_0\otimes D_0 & G[i]_{-1}\otimes D_1 & C'_1\otimes D_0 & C'_0\otimes D_1\\
        G[i]_1\otimes D_0 & \partial_{G,1}^\top \otimes \id_D & 0 & g[i]_1^{-1}\otimes \id_D & 0\\
        G[i]_0\otimes D_1 & \id_{G}\otimes \partial^\top_D & \partial^\top_{G,0}\otimes \id_D & 0 &g[i]_0^{-1}\otimes \id_D \\
        C'_1\otimes D_1 & 0 & 0 & \id_{C'}\otimes \partial^\top_D & {\partial'}^\top_{C} \otimes \id_D\\
    \end{pNiceMatrix},
    \end{align}
    and so,
    \begin{align}
        \delta^{\cone(f[i]),1}\begin{pmatrix}
            0\\
            \partial_{G,0}(g_0^{-1}[i](\bar{c}_m'))\otimes e_k\\
            0\\
            \bar{c}'_m\otimes e_k
        \end{pmatrix}&= \begin{pmatrix}
            0\\
            \partial_{G,0}^\top (\partial_{G,0}(g[i]^{-1}_0(\bar{c}_m'))\otimes e_k+g[i]_0^{-1} (\bar{c}'_m)\otimes e_k\\
            \partial^\top_{C'} (\bar{c}'_m)\otimes e_k
        \end{pmatrix}\\
        &=\begin{pmatrix}
            0\\
            g[i]^{-1}_0(\bar{c}_m')\otimes e_k+g[i]_0^{-1} (\bar{c}'_m)\otimes e_k\\
            \partial^\top_{C'} (\bar{c}'_m)\otimes e_k
        \end{pmatrix}=\vec{0}.
    \end{align}
    In the last equality, we use 
    $\partial^\top_C (\bar{c}_m')=\delta^C (\bar{c}_m')=0$ since $\bar{c}_m'\in \ker({\delta^C}')$.
    For fixed $\bar{c}'_m$, varying $k\in[n_d]$ in $\bar{X}^\mathrm{h,cone}_m$ yields logicals in the same homology class.
    Therefore, there are $n_d=d$ disjoint logicals given by Eq.~\eqref{eq:xh-disjoint-logicals} for different $k\in[d]$.
    Each $\bar{X}^\mathrm{h,cone}_m$ has a conjugate $Z$ logical, $\bar{Z}^{\mathrm{h,cone}}_m=(0,0,0,(e_m\otimes e_h)\otimes d)$ (see Remark~\ref{rmk:undeformed-z-cone}).
    Every $Z$ logical must anticommute with every representative of its conjugate $X$ logical.
    Since there are $d$ disjoint $\bar{X}^\mathrm{h}_m$ operators, $\bar{Z}^{\mathrm{h,cone}}_m$ must have weight at least $d$ in order to intersect each $\bar{X}^\mathrm{h}_m$ for $k\in[d]$ once.
\end{proof}

\begin{proposition}
    Logicals $\bar{Z}^\mathrm{v}=(0,0,c'_m\otimes e_l,0)\in H_1(\cone(f[i]))$ for $c'_m\in \ker(\partial_C')/\{b\}$ must have distance at least $d$.
    \label{prop:zv-distance}
\end{proposition}
\begin{proof}
    Take the trivial extension of the logical operators $\bar{X}^{\mathrm{v}}_i$,
    \begin{align}
        \bar{X}^{\mathrm{v},\cone}_i=(0,0,(b_i^\perp\otimes e_g)\otimes \bar{d},0),\quad \bar{d}\in \ker(\delta^D),
    \end{align}
    where $b_i^\perp$ is a basis vector for $\mathbb{F}_2^M/\ \mathrm{Im}(b)$.
    This operator cannot be the conjugate $X$ logical of measured logical $\bar{Z}_b^\mathrm{v} = (0,0,(b\otimes c)\otimes e_l,0)$, because they commute by direct application of Lemma~\ref{lem:orthogonal-basis}.
    In fact, $\bar{X}^{\mathrm{v},\cone}_i$ are the conjugate $X$ logicals of the $\bar{Z}^{\mathrm{v},\cone}_{i} = (0,0,(b^\perp_i\otimes c)\otimes e_l,0)$ logicals that are not being measured. 
    \begin{align}
        (\bar{X}^{\mathrm{v},\cone}_i)^\top \bar{Z}^{\mathrm{v},\cone}_{i'} &= 0+0+ (({b_i^\perp})^\top \otimes e_g^\top \otimes \bar{d}^\top ) (b^\perp_{i'}\otimes c\otimes e_l)+0\\
        &=({b_i^\perp})^\top b_{i'}^\perp\otimes e_g^\top c\otimes \bar{d}^\top e_l 
        =\delta_{i i'} 
    \end{align}
    Thus, we see that $\bar{Z}^{\mathrm{v},\cone}_{i}$ and $\bar{X}^{\mathrm{v},\cone}_i$ satisfy the correct commutation relations to form conjugate pairs.
    We can verify that $\bar{X}^{\mathrm{v},\cone}_i$ is indeed in $ H^1(\cone(f[i]))$, since
    \begin{align}
        \delta^{\cone(f[i]),1}\begin{pmatrix}
            0\\
            0\\
            (b_i^\perp \otimes e_g)\otimes \bar{d}\\
            0
        \end{pmatrix}=
        \begin{pmatrix}
            g[i]_1^{-1}(b_i^\perp\otimes e_g)\otimes \bar{d}\\
            0\\
            e_g\otimes \partial^\top_D(\bar{d})
        \end{pmatrix}=
        \vec{0}.
    \end{align}
    \noindent In the last equality, we use $g[i]_1^{-1}(b_i^\perp \otimes e_g)=0$ by assumption (Eq.~\eqref{eq:g1-condition}) and $\partial^\top_D(\bar{d})=\delta^D(\bar{d})=0$. We see that there are $n_C=d$ disjoint  $\bar{X}^{\mathrm{v},\cone}_i$ logicals for $g\in[n_C]$. For each $i\in [M-1]$,  $\bar{Z}^{\mathrm{v},\cone}_i$ must anticommute with all $d$ disjoint conjugate logicals $\bar{X}^{\mathrm{v},\cone}_i$. Therefore, $\bar{Z}^{\mathrm{v},\cone}_i$ must have weight at least $ d$.
\end{proof}
\noindent Together, propositions~\ref{prop:zh-distance} and~\ref{prop:zv-distance} imply that $d_1(\cone(f[i])) \ge d$.
Let us take an extended chain map over all ancilla systems in $t$ surgeries,
\begin{align}
    f'[i]:(\bigoplus_{i\in[t]} G[i]]\otimes D)_\bullet  \rightarrow (C'\otimes D)_\bullet.
\end{align}
Further, let 
\begin{align}
    \bar{f}[\tau] := \sum_{\substack{i\in[t]\\i\le \tau}} f'[i].
    \label{eq:iterative-map}
\end{align}
$\cone(\bar{f}[\tau])$ is the compacted code if we constructed it to include all surgeries up to $\tau\le t$. In other words, $\cc_\bullet = \cone(\bar{f}[t])$.
\begin{proposition}
    $d_1(\cone(\bar{f}))\ge d$.
\end{proposition}
\begin{proof}
     Propositions~\ref{prop:zh-distance} and~\ref{prop:zv-distance} apply to the extension of  $\cone(f[i])$ to $\cone(f'[i])$. To see this, consider extended logical operators of the toric code and the coned code.
    \begin{align}
        \bar{X'}^{\mathrm{h}}_m&=(0,\dots,0,\bar{c}'_m\otimes e_k)\in (\bigoplus_{t'\in[t]}(G[t']_0\otimes D_0,G[t']_{-1}\otimes D_1,))\oplus (C'_1\otimes D_0,C'_0\otimes D_1)\\
        \bar{X'}^{\mathrm{h},\cone}_m&=(0,\dots, 0,\partial_{G,0}(g[i]_{0}^{-1}(\bar{c}'_m))0,\dots,0,\bar{c}'_m\otimes e_k)\notag\\&\in (\bigoplus_{t'\in[t]}(G[t']_0\otimes D_0,G[t']_{-1}\otimes D_1,))\oplus (C'_1\otimes D_0,C'_0\otimes D_1)
        \label{eq:extended}
    \end{align}
    $\delta^{\cone(f'[i]]),1}$ only has non-zero entries that are shared with $\delta^{\cone(f[i]),1}$. 
    This implies that the extended $\bar{X}^{\mathrm{h},\cone}$ is still in $\ker(\delta^{\cone(f'[1]),1})$, and the rest of the arguments follow.
    We repeat this for the extension $\bar{X'}^{\mathrm{v},\cone}_i$ in proposition~\ref{prop:zv-distance}.

    Suppose that the compacted code consists of $t$ surgeries, each measuring a different, potentially entangling toric code $\bar{Z}^\mathrm{v}_{b_i}$ operator, $(b_i\otimes c\otimes e_l,0)$. 
    Define a basis $\{b_j^{\perp}\}$ for the orthogonal complement of $\mathrm{Span}\{b_i\}$.
    Consequently, $\bar{X}^{\mathrm{v}}_j=(b_j^\perp\otimes e_g\otimes e_l,0)$ are the $X$ logicals conjugate to the $\bar{Z}^\mathrm{v}_j$ operators that were never measured throughout $t$ surgeries.
    We will now construct $\cc_\bullet$ iteratively, considering just one surgery.
    Recall from Eq.~\eqref{eq:iterative-map}, we use the notation $\bar{f}[\tau]$ to denote the cumulative chain map when adding the chain maps from $\tau$ surgeries.
    By equation~\eqref{eq:iterative-map}, $\bar{f}[1]=f'[1]$.
    Note that $\{{\bar{X'}}^{\mathrm{h},\cone}_m\}$ and $\{\bar{X'}^{\mathrm{v},\cone}_j\}$ remain logicals of $\cone(f'[1])$, thus, the distances of the extended $Z$ logicals $\{\bar{Z'}^{\mathrm{h,\cone}}_m\}$ and $\{\bar{Z'}^{\mathrm{v},\cone}_j\}$  must be greater than $d$ from the arguments presented in propositions~\ref{prop:zh-distance} and~\ref{prop:zv-distance}.
    Now, we take $\tau=2$ and consider $\cone(\bar{f}[2])$. 
    $\delta^{\cone(f'[2]),1}$ shares non-zero entries with $\delta^{\cone(f[2]),1}$ and is $0$ otherwise.
    Therefore, $\delta^{\cone(\bar{f}[2])),1}=\delta^{\cone(f'[1]),1}\oplus \delta^{\cone(f'[2]),1}$ must only contain non-zero entries already in $\delta^{\cone(f[1]),1}$ and $\delta^{\cone(f[2]),1}$. 
    This implies that $\bar{X'}^{\mathrm{v},\cone}_i$ is in $ \ker(\delta^{\cone(\bar{f}[2])),1})$ and is still a logical of $\cone(\bar{f}[2]))$.
    However, as before, $\bar{X'}^{\mathrm{h},\cone}_m$ must gain support on $G[2]_{-1}\otimes D_1$ such that it is in $\ker(\delta^{\cone(\bar{f}[2])),1})$,
    \begin{align}
        \bar{X'}^{\mathrm{h},\cone}_m = (0,\partial_{G,0}(g[1]_0^{-1}(\bar{c}'_m))\otimes e_k,0,\partial_{G,0}(g[2]_0^{-1}(\bar{c}'_m))\otimes e_k,0,\dots,0,\bar{c}'_m\otimes e_k)\\
        \in (G[1]_0\otimes D_0, G[1]_{-1}\otimes D_1, G[2]_0\otimes D_0, G[2]_{-1}\otimes D_1,\dots,G[t]_{-1}\otimes D_1,)\oplus(C'_1\otimes D_0,C'_0\otimes D_1).
    \end{align}
    Nevertheless, there are still $d$ disjoint representatives of $\bar{X'}^{\mathrm{h},\cone}_m $ for varying $k\in [d]$.
    Therefore, the arguments in propositions~\ref{prop:zh-distance} and~\ref{prop:zv-distance} still apply, and $d_1(\cone(\bar{f}[2]))\ge d$.
    Iterate this argument for $\tau=3,4,...$ in $\cone(\bar{f}[\tau])$ until $\tau=t$. When $\tau=t$, $d_1(\cone(\bar{f}[t]))=d_1(\cc_\bullet)\ge d$ is proven.
\end{proof}
\end{document}